\documentclass[11pt]{article} 

\usepackage{fullpage}
\usepackage{bbm}
\usepackage{graphicx}
\usepackage{upref}
\usepackage{enumerate}
\usepackage{latexsym}
\usepackage{color,graphics}
\usepackage{comment} 
\usepackage{relsize}
\usepackage{url} 
\usepackage{bm} 
\usepackage{marvosym}

\usepackage[section]{algorithm} 
\usepackage{amsmath,amsthm,amssymb,algorithmic,epsfig,graphicx, tikz}

\usepackage{amsfonts}
\usepackage{varioref}
\usepackage[ansinew]{inputenc}
\usepackage{tikz}

\newtheorem{theorem}{Theorem}[section]
\newtheorem{Lemma}[theorem]{Lemma}
\newtheorem{claim}[theorem]{Claim}

\newtheorem{corollary}[theorem]{Corollary}
\newtheorem{fact}[theorem]{Fact}

\newtheorem{definition}[theorem]{Definition}

\newcommand{\R}{\mathbb{R}}

\newcommand{\E}{\mathbb{E}}

\DeclareBoldMathCommand\boldlangle{\left\langle} 
\DeclareBoldMathCommand\boldrangle{\right\rangle} 

\def\bra#1{\mathinner{\langle{#1}|}}
\def\ket#1{\mathinner{|{#1}\rangle}}
\newcommand{\braket}[2]{\langle #1|#2\rangle}

\renewcommand{\part}[2]{\frac{\partial #1}{\partial #2}}

\newcommand{\als}[1]{\begin{align*}#1\end{align*}}
\newcommand{\all}[2]{\begin{align}\label{#2} #1\end{align}}
\newcommand{\al}[1]{\begin{align} #1\end{align}}

\newcommand{\enum}[1]{\begin{enumerate}#1\end{enumerate}}

\newcommand{\en}[1]{\left ( #1 \right )}

\newcommand{\Span}{ \text{Span}}

\newcommand{\nl}{\notag \\}

\newcommand{\norm}[1]{\lVert#1\rVert}

\renewcommand{\E}{\mathbb{E}}

\newcommand{\thmref}[1]{\hyperref[#1]{{Theorem~\ref*{#1}}}}
\newcommand{\lemref}[1]{\hyperref[#1]{{Lemma ~\ref*{#1}}}}
\newcommand{\remref}[1]{\hyperref[#1]{{Remark~\ref*{#1}}}}
\newcommand{\corref}[1]{\hyperref[#1]{{Corollary~\ref*{#1}}}}
\newcommand{\eqnref}[1]{\hyperref[#1]{{Equation~(\ref*{#1})}}}
\newcommand{\claimref}[1]{\hyperref[#1]{{Claim~\ref*{#1}}}}
\newcommand{\remarkref}[1]{\hyperref[#1]{{Remark~\ref*{#1}}}}
\newcommand{\propref}[1]{\hyperref[#1]{{Proposition~\ref*{#1}}}}
\newcommand{\factref}[1]{\hyperref[#1]{{Fact~\ref*{#1}}}}
\newcommand{\defref}[1]{\hyperref[#1]{{Definition~\ref*{#1}}}}
\newcommand{\exampleref}[1]{\hyperref[#1]{{Example~\ref*{#1}}}}
\newcommand{\hypref}[1]{\hyperref[#1]{{Hypothesis~\ref*{#1}}}}
\newcommand{\secref}[1]{\hyperref[#1]{{Section~\ref*{#1}}}}
\newcommand{\chapref}[1]{\hyperref[#1]{{Chapter~\ref*{#1}}}}
\newcommand{\apref}[1]{\hyperref[#1]{{Appendix~\ref*{#1}}}}

\AtBeginDocument{}

\begin{document}
\title{A Quantum Interior Point Method for LPs and SDPs}

\author{ 
Iordanis Kerenidis \thanks{
CNRS, IRIF, Universit\'e Paris Diderot, Paris, France and  
Centre for Quantum Technologies, National University of Singapore, Singapore. 
Email: {\tt jkeren@irif.fr}.} 
\and
Anupam Prakash \thanks{CNRS, IRIF, Universit\'e Paris Diderot, Paris, France.
Email: { \tt anupamprakash1@gmail.com}.}
}

\maketitle 

\begin{abstract}
We present a quantum interior point method with worst case running time $\widetilde{O}(\frac{n^{2.5}}{\xi^{2}}  \mu \kappa^3 \log (1/\epsilon))$  for SDPs and 
$\widetilde{O}(\frac{n^{1.5}}{\xi^{2}}  \mu \kappa^3 \log (1/\epsilon))$ for LPs, where the output of our algorithm is a pair of matrices $(S,Y)$ 
that are $\epsilon$-optimal $\xi$-approximate SDP solutions. The factor $\mu$ is 
at most $\sqrt{2}n$ for SDPs and $\sqrt{2n}$ for LP's, and $\kappa$ is an upper bound on the condition number of the intermediate solution matrices. For the case 
where the intermediate matrices for the interior point method are well conditioned, our method provides a polynomial speedup over the best known classical SDP solvers and 
interior point based LP solvers, which have a worst case running time of $O(n^{6})$ and $O(n^{3.5})$ respectively. Our results build upon recently developed techniques for 
quantum linear algebra and pave the way for the development of quantum algorithms for a variety of applications in optimization and machine learning. 
\end{abstract}

\section{Introduction}

Semidefinite programming is an extremely useful and powerful technique that is used both in theory and in practice for convex optimization problems. It has also been widely used for providing approximate solutions to NP-hard problems, starting from the work of Goemans and Williamson \cite{GW95} for approximating the MAXCUT problem.   

A Semidefinite Program (SDP) is an optimization problem with inputs a vector $c \in \R^m$ and matrices $A^{(1)}, \ldots , A^{(m)},B$ in $\R^{n\times n}$, and the goal is to find a vector $x \in \R^m$ that minimizes the inner product $c^{t} x$, while at the same time the constraint $ \sum_{k \in [m]} x_{k} A^{(k)}  \succeq B$ is satisfied. In other words, an SDP is defined as

$$Opt(P) = \min_{x \in \R^{m}}  \{ c^{t} x \;|\; \sum_{k \in [m]} x_{k} A^{(k)}  \succeq B  \}.$$

\noindent One can also define the dual SDP as the following optimization problem,

\[ Opt(D)= \max_{Y \succeq 0 }  \{ Tr(BY) \;|\; Y \succeq 0, Tr(YA^{(j)})=c_{j} \}.  \]

\noindent For most cases of SDPs, strong duality implies that the two optimal values, $Opt(P)$ for the primal SDP and $Opt(D)$ for the dual SDP are actually equal. 
Semidefinite programs are a generalization of Linear Programs, which is the special case where all matrices are diagonal. The main advantage of both linear and semidefinite programs is that they encapsulate a large number of optimization problems and there are polynomial time algorithms for solving them. 

The ellipsoid algorithm \cite{K80} was the first provably polynomial-time algorithm to be given for LPs. In a seminal paper, Karmarkar \cite{K84} introduced an algorithm that is equivalent to an interior point method to improve the running time of the ellipsoid method. The first provably polynomial time algorithm for SDPs was given by Nesterov and Nemirovskii \cite{NN94} using interior point methods with self-concordant barrier functions.  

The running time of the best known general method for solving SDPs \cite{LWS15} is $O(m^{3} + mn^{\omega} + m^{2}ns \log (mn/\epsilon) )$ where $s$ is the sparsity, the maximum number of non zero entries in a row of the input matrices and the SDPs are solved to additive error $\epsilon$. The running time in the worst case is $\widetilde{O}(n^{6})$ for dense SDPs with $m=O(n^{2})$. For the case of linear number of constraint matrices ($m=O(n)$) the worst-case running time is $\widetilde{O}(n^4)$, or slightly better in case the sparsity of the constraint matrices is small.

For our results we will in fact use a version of the classical interior point method for SDPs based on \cite{BN15}, whose running time is
$$\tilde{O}(  ( n^{0.5} m^{3} + n^{2.5}m^{2} + n^{3.5}m) \log(1/\epsilon)).$$
For dense SDPs with $m=O(n^{2})$ the running time is $O(n^{6.5})$, while for instances with $m=O(n)$, the running time is $O(n^{4.5})$. This method can also be used for Linear Programming in $n$ dimensions with $m$ constraints, and running time $O(n^{2} (m+n)^{3/2} \log (1/\epsilon))$. 

The fastest known algorithm for linear programs is by Vaidya \cite{V89}, it has time complexity $O( (m+n)^{1.5} n L)$ where $L$ is the logarithm of the maximal sub-determinant of the constraint matrix. The main bottleneck for interior point style algorithms for both LPs and SDPs is maintaining suitable approximations of the inverse of the Hessian matrices at each step of the computation, Vaidya's algorithm uses a combination of pre-computation and low-rank updates to accelerate this step for linear programs. The interior point method was originally proposed for solving linear programs \cite{K84} and the best known LP algorithms using this approach \cite{V90} have complexity $O(m^{1.5}n^{2} +  m^{2}n )$ if all the 
arithmetic operations are performed up to a constant number of bits of precision. 

A different approach to SDP solving is the framework of Arora and Kale \cite{AK12, AHK05}. 
The Arora-Kale algorithm uses a variant of the multiplicative weights update method to iteratively find better solutions to the primal and the dual SDPs until a solution close to the optimal is found. The running time of the algorithm depends on the dimensions $m,n$ of the problem, the approximation guarantee $\epsilon$ between the solution and the optimal value, and upper bounds $R$ and $r$ on the "size" of the optimal primal and dual solutions (in some appropriate norm).  More precisely, an upper bound on the running time of this algorithm given in \cite{AGGW17} is, 
\[ \tilde{O} \big( nms \left( \frac{Rr}{\epsilon}  \right)^4 +ns\left( \frac{Rr}{\epsilon}  \right)^7 \big). \]
While the Arora-Kale framework has found many applications in complexity theory \cite{AK12}, for the case of solving general SDPs, it remains an algorithm of mostly theoretical interest. It is known that for many combinatorial problems (like MAXCUT or scheduling problems), the width $(Rr/\epsilon)$ grows at least linearly in the dimensions $n, m$ (Theorem 24, \cite{AGGW17}), making this algorithm infeasible in practice and with worse running time than the interior point method. 

Recently, quantum algorithms for solving semi definite programs (SDPs) using the Arora-Kale framework were first proposed by Brandao and Svore \cite{BS16} and subsequently improved by van Appeldoorn, Gilyen, Gribling and de Wolf \cite{AGGW17}. These algorithms had a better dependence on $n,m$ but a worse dependence 
on other parameters compared to the classical algorithm. Very recently, \cite{BKLLSW17} and independently \cite{AG18} (subsequent to a previous version of \cite{BKLLSW17}) provided an even better dependence
on the parameters $n,m$ and improved the dependence on the error to $(\frac{Rr}{\epsilon})^{4}$. In order to discuss these results and compare them to ours, we first need to describe the quantum input models for SDPs. 

The most basic input model is the sparse oracle model of $ \cite{BS16}, \cite{AGGW17}$, where the input matrices $A^{(i)}$ are assumed to be $s$-sparse and one has access to an oracle $O_{A}: \ket{i, k, l,0} \to \ket{i,k,l, index(i,k,l)}$ for computing the $l$-th element of $A^{(i)}_{k}$. The quantum state model was introduced in $\cite{BKLLSW17}$, and in this model each $A^{(i)} $ is split as a difference of positive semidefinite matrices ($A^{(i)} = A^{(i)}_{+} - A^{(i)}_{-}$) and we assume that we have access to purifications of the density matrices corresponding to $A^{(i)}_{+} , A^{(i)}_{-}$ for all $i \in [m]$. 

The input model most relevant for our work is the quantum operator model of \cite{AG18} where one assumes access to unitary block encodings of the the input matrices $A^{(i)}$, that is there are efficient implementations of unitary operators $U_{j}= \left ( \begin{matrix} 
A^{(j)}/\alpha_{j} & .  \\
. & .  \\
 \end{matrix} \right ) $. That is, we assume that the operations $\ket{j}\ket{ \phi} = \ket{j} U_{j} \ket{\phi} $ can be performed efficiently.  It is shown in \cite{AG18} that both the sparse oracle model and the 
 quantum state model can be viewed as instances of the operator model. The best upper bound for the quantum algorithm from \cite{AG18} in the operator model is, 
\all{ 
 \tilde{O} \en{ \en{ \sqrt{m} + \sqrt{n} \en{ \frac{Rr}{\epsilon} } }  \en{ \frac{Rr}{\epsilon}}^{4} \alpha}. 
}{optsdp}  
We note that this parameter $\alpha$ can be $\sqrt{n}$ in the worst case but it can in principle be smaller than that. The various input models for quantum SDP solvers assume oracle access to the input matrices in the form described above and do not address further the question 
of implementing these oracles. In this paper, we work in the quantum data structure model introduced in \cite{KP16, KP17} 
that explicitly provides a method for implementing block encodings for arbitrary matrices.

In the quantum data structure model the algorithms have quantum access to a data structure that stores the matrix $A^{(i)}, i \in [m]$ in a QRAM (Quantum Random Access Memory). The data structure is built in a single pass over a stream of matrix entries $(i, j, a_{ij})$ and the time required to process a single entry is poly-logarithmic in $n$. Thus, the construction of the data structure does not incur additional overhead over that required for storing the matrices $A^{(i)}$ sequentially into classical memory or the QRAM. 
More generally, we can define the quantum data structure model for a general data structure $D$ as follows. 
\begin{definition} \label{defer} 
A data structure for storing a dataset $D$ of size $N$ in the QRAM is said to be efficient if it can be constructed in a single pass over the entries $(i, d_{i})$ for $i \in [N]$ 
and the insertion and update time per entry is $O(\log^{2} N)$. 
\end{definition}

In \cite{KP16, KP17} we had given quantum algorithms for matrix multiplication and inversion for an arbitrary matrix $A \in \R^{n\times n}$ with $\norm{A}=1$ in the QRAM data structure model. The running time 
for these algorithms was $\widetilde{O}( \mu(A) \kappa^{2}(A)/\epsilon)$ where $\epsilon$ is the error, $\kappa(A)$ is the condition number, and $\mu(A)=  \min_{p\in [0,1]} (\norm{A}_{F}, \sqrt{s_{2p}(A)s_{(1-2p)}(A^{T})})$ for $s_{p}(A) = \max_{i \in [n]} \sum_{j \in [n]} A_{ij}^{p}$. Note that $\mu(A) \leq \norm{A}_{F} \leq \sqrt{n}$ as we have assumed that $\norm{A}=1$. 

The different terms in the minimum in the definition of $\mu(A)$ correspond to different choices for the data structure for storing $A$. Subsequently, 
\cite{CGJ18} pointed out that these results are equivalent to there being an efficient block encoding for $A$ with parameter $\alpha=\mu(A)$. Moreover, in \cite{CGJ18}, the dependence on precision for the quantum linear system solvers in the QRAM data structure model was improved to $\text{polylog}(1/\epsilon)$ and that on $\kappa$ to linear.
The QRAM data structure model thus explicitly provides efficient block encodings for matrix $A$ with $\alpha = \mu(A)$ and can be used to implement the operator model 
for SDPs as described above. 
 
The quantum interior point algorithm is developed in the QRAM data structure model and can therefore be compared to the best quantum SDP algorithm in the Arora-Kale framework \cite{AG18} with running time 
given by equation \eqref{optsdp}. We note that both $\alpha$ and $\max_{i \in [m]} \mu(A^{(i)})$ are at most $s$ when the input SDP matrices are $s$-sparse, so the algorithms 
remain comparable for the case of sparse SDPs as well. There is a linear lower bound on $Rr/\epsilon$ (Theorem 24, \cite{AGGW17}), and hence in the best case, i.e. when $Rr/\epsilon=O(n)$, the running time in equation \eqref{optsdp} is $\widetilde{O}(n^{6})$, which is similar to the running time for the classical interior point method.

\subsection{Our techniques and results}
In this paper we develop the first quantum algorithms for SDPs and LPs using interior point methods and obtain a significant polynomial speedup over the corresponding classical algorithms if 
the intermediate 
matrices arising in the interior point method are well-conditioned. 

The classical interior point method starts with a pair of feasible solutions to the SDP and iteratively updates the solutions, improving the duality gap in each iteration. 
Each iteration of the interior point method involves constructing and solving the Newton linear system whose solutions give the updates to be applied to the SDP 
solution. 

The main bottleneck for our quantum interior point method, as for the classical interior point method, is that the matrix for the Newton linear system is not given directly and is expensive to compute from the data. One of our main contributions is a technique for constructing block encodings for the Newton linear system matrix which allows us to solve this linear system with low cost in the quantum setting. We utilize the quantum techniques for linear algebra developed in \cite{KP16,KP17} and the improvements in precision for these linear algebra techniques in \cite{CGJ18, GLSW18}.

The quantum solution of the Newton linear system corresponds to a single step of the interior point method. We need a classical description of the solution of the linear system in order to define the linear system for the next step. We therefore perform tomography on the quantum state corresponding to the solution of the Newton linear system at each step of the method. The dimension of the state that we perform tomography on is $d=O(n^{2})$ so a super-linear tomography algorithm \cite{OW16, GLFBE10} would be prohibitively expensive. We provide a tomography algorithm that given a unitary for preparing a vector state, reconstructs a vector 
$\delta$-close to the target vector in the $\ell_{2}$ norm with complexity $O(\frac{d \log d}{\delta^{2}})$ which is linear in the dimension.

 We also prove a convergence theorem in order to upper bound the number of iterations of the method, taking into account the various errors introduced by the quantum algorithms. The algorithm outputs a pair of matrices $(S,Y)$ that are $\epsilon$-optimal $\xi$-approximate SDP solutions, where by $\epsilon$-optimal we mean that forthe duality gap we have $Tr(SY) < \epsilon$ and by $\xi$-approximate we mean that the SDP constraints are satisfied to error $\xi$. The classical analysis also implies that the algorithm converges if the precision for the tomography algorithm is $\delta= O(1/\kappa)$ where $\kappa$ is the condition number of the matrix $Y$ being updated.

We will provide the exact running time of our quantum interior point method in a later section after defining precisely the matrix of the Newton Linear system, but we describe it here without making explicit the matrices involved and suppressing logarithmic factors. Our interior point algorithm for SDPs has worst case running time, 
\[
\widetilde{O}(\frac{n^{2.5}}{\xi^{2}}  \mu \kappa^3 \log 1/\epsilon)
\]

\noindent The number of iterations $T$ for convergence is $\widetilde{O}(\sqrt{n})$, this is same as in the classical case. The term $\widetilde{O}(n^2 \kappa^2/\xi^2)$ comes from the tomography algorithm that reconstructs the solution to the Newton linear system in each step to $\ell_{2}$ error $\xi/\kappa$. The additional $\mu$ and $\kappa$ factors arise from the quantum linear system solver as described earlier. Note that since the Newton linear system matrix has dimension $(n^2+m) \times (n^2+m)$ the $\mu$ factor is at most $\sqrt{2}n$. 

lf the intermediate matrices arising in the method are well-conditioned, the worst case running time for our algorithm is $\widetilde{O}(n^{3.5})$ which is significantly better than the corresponding classical interior point method in \cite{BN15} whose running time is $O(n^{6.5})$ and the best known general SDP solver \cite{LWS15} whose running time is $O(n^6)$. 

While we have a worst case bound of $n$ on $\mu$, this parameter can be much smaller in practice than the worst case scenario. For example, in \cite{KL18}, a quantum linear system solver was used for dimensionality reduction and classification of the MNIST data set, and in that case the dimension of the corresponding matrix was $10^5$ while the value of $\mu$, when taken to be the maximum $\ell_1$ norm of any row, was found to be less than 10. Another remark is that the real condition number can be replaced by a smaller condition threshold. 

In this paper we focus on the case of dense SDPs where $m=O(n^{2})$, for the sparse case when $m=O(n)$ as in the classical case one can develop variants of the interior point method with faster running time. It is easy to see that these methods can be quantized, we do not address these methods for the sparse case in this work. 

For the special case of Linear Programs our algorithm has running time
\[
\widetilde{O}(\frac{n^{1.5}}{\xi^{2}}  \mu \kappa^3 \log (1/\epsilon))
\]
The condition number $\kappa$ here is a bit different than the SDP case, it is the ratio of the maximum to the minimum element of the intermediate solution vectors, while $\mu$ is in the worst case at most $\sqrt{2n}$.
The running time of our algorithm is $\widetilde{O}(n^2)$ if the intermediate Newton matrices and solution vectors are well conditioned, compared to $O(n^{3.5})$ for the corresponding classical algorithm in \cite{BN15} and \cite{V90}. We note that there 
is a better specialized classical algorithm for linear programs\cite{V89} with time complexity $O( (m+n)^{1.5} n L)$ where $L$ is the logarithm of the maximal sub-determinant of the constraint matrix.

Our results provide the first quantum SDP and LP solvers based on the interior point method, thus paving the way for a vast variety of applications. 
An interesting feature of our algorithm is that the quantum part is no more difficult than a quantum linear system solver, a circuit whose depth depends on $\mu$ and $\kappa$ but only logarithmically on the dimension and the error parameter. This quantum circuit is repeated independently at each step of the iterative method, and for each step a number of times required by the tomography. Therefore, building a quantum circuit for solving linear systems implies a quantum SDP solver.

 This paper is organized as follows. In Section \ref{prelim} we collect some linear algebra preliminaries and useful quantum procedures. The results on quantum linear system  solvers and quantum data structures that we need are introduced in Section \ref{elsa}. In Section \ref{tom} we provide an algorithm for vector state tomography. In 
Section \ref{cipm} we describe the classical interior point method and in \ref{clan} we provide a convergence analysis for the interior point method when the linear systems are solved approximately. We present the quantum interior point method for SDPs and LPs in Section \ref{qipm} and bound its running time. 

\section{Preliminaries} \label{prelim} 
We introduce some notation that is used throughout the paper. The entry-wise or Hadamard product of two matrices $A,B \in \R^{m\times n}$ is denoted by $A \odot B$. The concatenation of two vectors $a \in \R^{n}, b \in \R^{m}$ is denoted by $a \circ b$. The vector state $\ket{a}$ for $a \in \R^{n}$ is the quantum state $\frac{1}{ \norm{a}} \sum_{i \in [n]} a_{i} \ket{i}$. Given a vector $a \in \R^{n}$ the vectors obtained by taking the entry-wise square root and squares of $a$ are denoted as $\sqrt{a}$ and $a^{2}$ respectively. If $A, B \in \R^{n\times n}$, then $A \oplus B \in \R^{2n \times 2n}$ denotes the block diagonal matrix $\left ( \begin{matrix} 
A & 0 \\ 
0 & B \end{matrix} \right )$. 

The singular value decomposition (SVD) for a matrix $A \in \R^{m\times n}$ is denoted as $A= \sum_{i} \sigma_{i} u_{i} v_{i}^{t}$. The spectral decomposition  for a symmetric matrix $A \in \R^{n\times n}$ is denoted as $A= \sum_{i \in [n]} \lambda_{i} u_{i}u_{i}^{t}$. The condition number $\kappa(A):= \frac{\sigma_{max}(A)}{\sigma_{min}(A)}$ of an invertible matrix is the ratio between the largest and smallest singular values of $A$. The Frobenius norm $\norm{A}_{F} = \en{ \sum_{i} \sigma_{i}^{2} }^{1/2}$ and the spectral norm $\norm{A} = \sigma_{max}(A)$ are functions of the singular values. The $\ell_{2}$ norm of a vector $v$ is denoted as $\norm{v}$. The $i$-th column of $A$ is denoted by $A_{i}$ and the $j$-th row is denoted as $A^{j}$.

\subsection{Linear Algebra} 
We collect in this section linear algebra facts and lemmas that are used in later sections.

\begin{fact} \label{isospec} 
For all $A, B \in \R^{n \times n}$, $\norm{AB}_{F} = \norm{BA}_{F}$. 
\end{fact}


\begin{fact} \label{frobfact} 
\noindent $\lambda_{min}(A) \norm{B}_{F} \leq \norm{AB}_{F} \leq \lambda_{max}(A)  \norm{B}_{F}$ for positive definite $A \succ 0$ and $B \in \R^{n\times n}$. 
\end{fact} 
\noindent A simple lemma that will be used later is stated below. 
\begin{Lemma} \label{frobineq} 
Let $B, Y, Y' \in \R^{n\times n}$ be such that $Y' \preceq (1+ \rho) Y$, then $\norm{Y'B}_{F} \leq (1+\rho) \norm{YB}_{F}$. 
\end{Lemma} 
\begin{proof} 
The squared Frobenius norm of $Y'B$ is the sum of the squared norms of the columns of $Y'B$. 
Denoting the columns of $B$ by $b_{i}, i \in [n]$, we have, 
\al{ 
\norm{Y'B}_{F}^{2} = \sum_{i \in [n]} \norm{ Y'b_{i} } ^{2} \leq (1+ \rho)^{2}  \sum_{i \in [n]} \norm{ Yb_{i} } ^{2} = (1+ \rho)^{2} \norm{YB}_{F}^{2} 
} 
where the inequality follows as $Y' \preceq (1+\rho )Y$.  
\end{proof} 

\subsection{Quantum procedures}


We will require the following auxiliary lemma that gives a procedure for 
preparing $\ket{ a \circ b}$ given procedures for preparing $\ket{a}$ and $\ket{b}$. This lemma will be useful for preparing block encodings required for the quantum interior point method. 

\begin{Lemma} \label{concat} 
Let $a \in \R^{m}, b \in \R^{n}$ for integers $m,n>0$ and let $U_{a}: \ket{0^{\lceil \log m \rceil} } \to \ket{a}$ and $U_{b}: \ket{0^{\lceil \log n \rceil}  } \to \ket{b}$ be unitary operators that 
can be implemented in time $T(U_{a}), T(U_{b})$ respectively, then given $\norm{a}, \norm{b}$ 
the state $\ket{a \circ b}$ 
can be prepared in time $O(T(U_{a}) + T(U_{b}))$. 
\end{Lemma} 
\begin{proof} 
We start with an auxiliary qubit in the 
state $\frac{\norm{a}}{\norm{a\circ b}} \ket{0} + \frac{\norm{b}}{\norm{a\circ b}} \ket{1}$ and an output register initialized to $\ket{ 0^{\lceil \log (m) \rceil + \lceil \log (n) \rceil} }$, 
where $\norm{ a \circ b} = (\norm{a}^{2} + \norm{b}^{2})^{1/2}$. The initial state is, 
\als{ 
\frac{\norm{a}}{\norm{a\circ b}} \ket{0,0^{  \lceil \log (m) \rceil + \lceil \log (n) \rceil}    } + \frac{\norm{b}}{\norm{a\circ b}} \ket{1,0^{\lceil \log (m) \rceil + \lceil \log (n) \rceil}   } 
} 
Conditioned on the control qubit being $\ket{0}$ apply $(I \otimes U_{a} )$ to obtain the state $ (I \otimes U_{a} ) \ket{ 0^{\lceil \log n \rceil}} \ket {0^{\lceil \log m \rceil} } = \ket {0^{\lceil \log n \rceil} } \ket{a}  $. 
Conditioned on the control qubit being $\ket{1}$ apply $(I \otimes U_{b})$ to obtain $(I \otimes U_{b}) \ket{0^{\lceil \log m \rceil}} \ket {0^{\lceil \log n \rceil} } =  \ket {0^{\lceil \log m \rceil} }  \ket{b}$. After this step we obtain the state, 
\als{ 
\frac{1}{\norm{a\circ b}}  \en{ \ket{0} \sum_{k\in [m]}  a_{k} \ket{k}  +\ket{1} \sum_{l \in [n] }   b_{l} \ket{l} } 
} 
Let $q=2^{\lceil \log (m) \rceil + \lceil \log (n) \rceil}$ be the dimension of the output register.  Define the unitary $V$ which acts as follows: $V\ket{0}\ket{t} = \ket{0}\ket{t}$ for $t \in [q]$ and $V\ket{1}\ket{t} = \ket{1}\ket{t+m \mod (q) }$ for $b\in [q]$. 
Note that $V$ is a unitary as it permutes the orthogonal basis states $\ket{0,b}, \ket{1,b}$ for $t\in [2q]$. Applying $V$ to the above state we obtain, 
\als{ 
\frac{1}{\norm{a\circ b}}  \en{ \ket{0} \sum_{k\in [m]}  a_{k} \ket{k}  +\ket{1} \sum_{l \in [n] }   b_{l} \ket{l+m} } 
} 
Apply a bit-flip conditioned on the auxiliary register being in a state $\ket{t'}$ such that $t'>m$, this operation erases the control qubit as it maps $\ket{1} \ket{l+m} \to \ket{0} \ket{l+m}$. 
Discarding the auxiliary qubit, after this operation we have the desired state, 
\als{ 
\frac{1}{\norm{a\circ b}} \en{ \sum_{k\in [m]}  a_{k} \ket{k}  + \sum_{l \in [n] }  b_{l} \ket{l+m} }  = \ket{a \circ b} 
} 
The time required for the procedure is $T(U_{a}) + T(U_{b})$ for the conditional operations, the remaining steps have negligible cost. 
\end{proof} 

\noindent We also state an approximate version of the above lemma where the norms $\norm{a}, \norm{b}$ are estimated within error $(1\pm\delta)$. 
\begin{corollary} 
Given estimates $\overline{\norm{a}} \in (1\pm \delta) \norm{a}, \overline{\norm{b}} \in (1\pm \delta) \norm{b}$ and unitaries $U_{a}, U_{b}$ in Lemma \ref{concat} a state $\ket{z}$ such that $\norm{\ket{z} - \ket{a \circ b} }_{2} \leq 2\delta$ can be prepared in time  $O(T(U_{a}) + T(U_{b}))$. 
\end{corollary} 
\noindent The proof is straightforward and follows by replacing $\norm{a}, \norm{b}$ in the proof of Lemma \ref{concat} by the respective estimates.

\section{Quantum linear system solvers and QRAM data structures} \label{elsa} 
In this section we collect the results on quantum linear algebra primitives and QRAM data structures that are required for the interior point method.

\subsection{Quantum linear system solvers} 
We assume without loss of generality that the matrices $A \in \R^{n \times n}$ are symmetric. If the matrix $A$ is rectangular or not symmetric, then it is well known we can 
work instead with the symmetrized matrix 
$ \overline{A} = \left ( \begin{matrix}
0,  &A \\ 
A^{T} , &0 \\
\end{matrix}  \right )$.

In \cite{KP16, KP17} we constructed efficient data structures (Definition \ref{defer}) for storing matrices $A$ that were used to obtain algorithms for quantum linear algebra operations including matrix inversion and multiplication with running time $O(\mu(A) \kappa^{2}(A)/\epsilon)$ where 
$\mu(A) \leq \sqrt{n}$ is a factor that depends on the matrix $A \in \R^{n\times n}$ and $\epsilon$ is the accuracy in the $\ell_{2}$ norm. These algorithms used phase estimation and amplitude amplification and therefore require time inverse polynomial in the precision and quadratic in the condition number $\kappa$. 
In recent work \cite{CGJ18}, the dependence on precision for the quantum linear system solvers in the QRAM data structure model was improved to $\text{polylog}(1/\epsilon)$ and that on $\kappa$ to linear. In order to state the improved results, we 
recall the notion of a $(\mu, t, \delta)$ block encoding of a matrix introduced in \cite{CGJ18}. 
\begin{definition} 
Let $A \in \R^{n\times n}$ be a $\lceil \log n \rceil$-qubit operator, an $(\lceil \log n \rceil+t)$ qubit unitary $U$ is an $(\mu, t, \delta)$ block encoding of $A$ if 
$U= \left ( \begin{matrix} 
\widetilde{A} & .  \\
. & .  \\
\end{matrix} 
\right ) $ such that $\norm{ \mu \widetilde{A} - A } \leq \delta$. A $(\mu, t, \delta)$ block encoding of $A$ is said to be efficient if it can be implemented in time $T_{U}=O(\text{polylog}(n))$. 
\end{definition}  
\noindent We next state the results on improved linear system solvers and matrix multiplication from \cite{CGJ18, GLSW18}. In the theorem stated below, the first part is given as Lemma 27, the second part of the result follows from Lemma 22 and the third part follows from Corollary 29 and Theorem 21 in \cite{CGJ18}. The singular value transformation
approach \cite{GLSW18} can also be used to obtain these results. 
\begin{theorem} \label{qlsa} 
\cite{CGJ18, GLSW18} Let $A \in \R^{n\times n}$ be a matrix with non-zero eigenvalues in the interval $[-1, -1/\kappa] \cup [1, 1/\kappa]$. Given an implementation of an $(\mu, O(\log n), \delta )$ block encoding for $A$ in time $T_{U}$ and a procedure for preparing state $\ket{b}$ in time $T_{b}$, 
\enum{ 
\item If $\delta \leq \frac{\epsilon}{\kappa^{2} \text{polylog} (\kappa, 1/\epsilon)}$ then a state $\epsilon$-close to $\ket{A^{-1} b}$ can be generated in time 
$O((T_{U} \kappa \mu+ T_{b} \kappa) \text{polylog}(\kappa \mu /\epsilon))$. 
\item If $\delta \leq \frac{\epsilon}{2\kappa}$ then a state $\epsilon$-close to $\ket{A b}$ can be generated in time $O((T_{U} \kappa \mu+ T_{b} \kappa) \text{polylog}(\kappa \mu /\epsilon))$. 
\item For $\epsilon>0$ and $\delta$ as in parts 1 and 2 and $\mathcal{A} \in \{ A, A^{-1} \}$, an estimate $\Lambda$ such that $\Lambda \in (1\pm \epsilon) \norm{ \mathcal{A} b}$ with probability $(1-\delta)$ can be generated in time $O((T_{U}+T_{b}) \frac{ \kappa \mu}{ \epsilon} \text{polylog}(\kappa \mu/\epsilon))$. 
} 
\end{theorem} 
\noindent We require the following simple result on composing such operations.
\begin{theorem} \label{qlsa1}
\cite{CGJ18, GLSW18}  
Given $(\mu(M_{i}), O(\log (n)), \delta_{i})$-block encoding for matrices $M_{i}$ implemented in time $T_{i}$ for $i \in \{1,2,3\}$ such that 
$\delta_{i} \leq \widetilde{O}(\epsilon/\kappa(M_{i})^{2})$ and state $\ket{b}$, a state $\epsilon$-close to the $\ket{\mathcal{M} b}$ for $\mathcal{M}= \prod_i \mathcal{M}_{i} $ with $\mathcal{M}_{i} \in \{ M_{i},M_{i}^{-1} \}$ can be 
generated in time 
$\widetilde{O}(\kappa(\mathcal{M}) (\sum_{i} \mu(\mathcal{M}_{i}) ) (\sum_{i} T_{i}) \log (1/\epsilon))$ and an $\epsilon$-estimate to the norm $\norm{\mathcal{M} b}$ in time 
$\widetilde{O}(\frac{1}{\epsilon}\kappa(\mathcal{M}) (\sum_{i} \mu(\mathcal{M}_{i}) ) (\sum_{i } T_{i}) \log (1/\epsilon))$.  
\end{theorem} 
\noindent The running time has a factor of $\kappa(\mathcal{M})$ instead of a factor $\prod_i \kappa(\mathcal{M}_{i})$ since we do not perform amplitude amplification after every matrix operation but only at the end. 


It is standard to assume for quantum linear system solvers that the eigenvalues of $A$ belong to the interval $[-1, -1/\kappa] \cup [1, 1/\kappa]$. 
The quantum data structures constructed in \cite{KP16, KP17} implement a unitary $U$ which is a 
$(\mu(A), \log(n), 0)$ block encoding for a matrix $A$ in time $T_{U}=O(\text{polylog}(n))$. 

\begin{theorem}  \label{ds} 
\cite{KP17, KP16} There are efficient QRAM data structures for storing vectors $v_{i} \in \R^{n}, i \in [m]$ and matrices $A\in \R^{n\times n}$ such that with access 
to these data structures there is, 
\enum{ 
\item  An $\widetilde{O}(1)$ time state preparation procedure $\ket{i} \ket{0} \to \ket{i} \ket{ v_{i}}$ for $i \in [m]$. 
\item  An efficient $(\mu(A), O(\log n), 0)$ block encoding for $A$ for \\ $\mu(A)= \min_{p\in [0,1]} (\norm{A}_{F}, \sqrt{s_{2p}(A)s_{(1-2p)}(A^{T})})$. 
} 
\end{theorem}

We also require an auxiliary result which provides error bounds for implementing a block encoding for a matrix whose rows and columns can be prepared within $\ell_{2}$ error $\epsilon$. It follows from the second part of Lemma 23 in \cite{CGJ18} which provides an analysis for the QRAM data structure in \cite{KP17} for the approximate case.  
\begin{Lemma} \label{lepsilon} 
\cite{CGJ18, KP17} Let $A \in \R^{n\times n}$ be a matrix such that the unitaries $U: \ket{i, 0} \to \ket{i, A_{i}}$ and $V: \ket{0, j} \to \frac{1} { \norm{A}_{F} } \sum \norm{A_{i}} \ket{i,j}$ can be implemented to error $\epsilon$ in time $T_{A}$, then an $(\norm{A}_{F}, O(\log n), \epsilon)$ block encoding for $A$ can be implemented in time $T_{A}$. 
\end{Lemma}

\subsection{Quantum data structures} 
For the quantum interior point method we require that the input matrices $A^{(k)}, k \in [m]$ and $B$ for the SDP input are stored in appropriate QRAM data structures so that we can perform multiplication or inversion with these matrices with the guarantees in Theorem \ref{qlsa}. 

Let $\mathcal{A}_{i,j} \in \R^{m}$ be the vector with entries 
$\mathcal{A}_{i,j,k} = (A^{(k)})_{i,j}$ for $k \in [m]$. We also need to prepare the states $\ket{\mathcal{A}_{i,j}}$ for $i, j \in [n]$. 
The SDP input matrices $A^{(k)}, k \in [m]$ are stored as a tensor $T\in \R^{ n \times n \times m} $. We provide efficient QRAM data structures that enable 
the above operations to be carried out in time $\widetilde{O}(1)$. We show next that this can be done using the same algorithms as in \cite{KP17, KP16} and with the 
same running time guarantees. 

Let $T \in \R^{n \times n \times m}$ be a 3-dimensional tensor. Let $\mathcal{S}(T):= \{ T_{:, :, k} \;|\; k \in [m] \}$ be the set of matrices that can be obtained by fixing the 
third index of $T$ and let $V(T) := \{ T_{i, j, :} \;|\; i, j \in [m] \}$ be the set of vectors that can be obtained by fixing the first two indices. We have the following corollary extending Theorem \ref{ds} for the storage of tensors. 

\begin{corollary} \label{cds} 
There are efficient QRAM data structures for storing the tensor $T\in \R^{n \times n \times m}$ such that given quantum access to these data structures, there are $(\mu(M), O(\log n), \epsilon)$ efficient block encodings for all $M \in \mathcal{S}(T)$ for $\mu(M)=  \min_{p\in [0,1]} (\norm{M}_{F}, \sqrt{s_{2p}(M)s_{(1-2p)}(M^{T})})$ and $\widetilde{O}(1)$ time preparation procedures $\ket{i, j, 0} \to \ket{i, j, v_{ij}}$ for all $v_{ij} \in V(T)$. 
\end{corollary}
\begin{proof} 

The algorithm maintains separate data structures for all matrices in $\mathcal{S}(T)$ and vectors in $V(T)$. 
A single entry $(T_{ijk}, i,j,k)$ belongs to exactly one matrix in $S(T)$ and one vector in $V(T)$. On receiving the entry $(T_{ijk}, i,j,k)$ the algorithm updates 
the two data structures corresponding $M_{k}$ and $v_{ij}$. This incurs a constant overhead 
over Theorem \ref{ds} and thus achieves the same running time and guarantees. 
\end{proof}

\section{Tomography for efficient vector states} \label{tom} 

We present in this section an algorithm for tomography of pure states with real amplitudes when we have an efficient unitary to prepare the states. The tomography algorithm 
will be used to recover classical information from the $d=O(n^{2})$ dimensional states corresponding to the solutions of the Newton linear system in each step of the interior point method. 

Let $U$ be a unitary operator that creates copies of the vector state $\ket{x}$ for a unit vector $x\in \R^{d}$ in time $T_{U}$. We also assume that we can apply the controlled version of $U$ in the same time. Our tomography algorithm outputs a unit vector $\widetilde{x}$ such that $\norm{  \widetilde{x} - x}_{2}  < \sqrt{7} \delta$ with probability $1-1/poly(d)$. The algorithm uses $N=\widetilde{O}(d/\delta^{2})$ calls to $U$ and runs in time $O(T_{U}N)$, thus the sample and time complexity for our algorithm are both $\widetilde{O}(d/\delta^{2})$ when $T_U$ is polylogarithmic.  Such a  linear time tomography algorithm for real pure states does not follow from existing results. 

A sample-efficient tomography algorithm for mixed states that requires $O(d/\delta^{2})$ copies of $\rho$ to 
obtain estimate $\rho'$ with the guarantee that $\norm{ \rho - \rho'}_{F} \leq \delta$ was given in O'Donnell and Wright \cite{OW16}. The time complexity of this approach is high as it uses a computationally expensive state estimation procedure \cite{K06}. Another approach to tomography is through compressed sensing \cite{GLFBE10}. While this is sample-efficient for low-rank (including pure states), it is also computationally expensive as it solves an SDP to perform the reconstruction. 

We note that our tomography problem is easier than the more general problems addressed in \cite{OW16, GLFBE10}, since we assume we can apply the unitary that prepares the state and its controlled version. Note also that for our purposes we only need to deal with real amplitudes, though complex amplitudes could be dealt with in a similar way. 
We next state our tomography algorithm and 
establish its correctness. 

\begin{algorithm} [H]
\caption{Vector state tomography algorithm. } \label{alg:tom} 
\begin{algorithmic}[1]
\REQUIRE Access to a unitary $U$ such that $U\ket{0} = \ket{x}= \sum_{i \in [d]} x_{i} \ket{ i}$ and to its controlled version. 
\begin{enumerate}
\item {\bf Amplitude estimation}
\begin{enumerate} 
\item Measure $N= \frac{36 d \ln d}{\delta^{2}}$ copies of $\ket{x}$ in the standard basis and obtain estimates $p_{i} = \frac{n_{i}}{N}$ where $n_{i}$ 
is the number of times outcome $i$ is observed. 
\item Store $\sqrt{p_{i}}, i \in [d]$ in QRAM data structure so that $\ket{p} = \sum_{i \in [d]} \sqrt{p_{i}} \ket{i}$ can be prepared efficiently. 
\end{enumerate}
\item {\bf Sign estimation}
\begin{enumerate}
\item Create $N=\frac{36 n \ln n}{\delta^{2}}$ copies of the state $ \frac{1}{ \sqrt{ 2}} \ket{0} \sum_{i \in [d]} x_{i} \ket{ i} +  \frac{1}{ \sqrt{ 2}} \ket{1} \sum_{i \in [d]}  \sqrt{p_{i}} \ket{i}$ using a control qubit. 
\item Apply a Hadamard gate on the first qubit of each copy of the state to obtain $\frac{1} { 2 }  \sum_{i \in [d]} [ (x_{i} + \sqrt{p_{i}}) \ket{0, i} + (x_{i} - \sqrt{p_{i}}) \ket{1, i} ].$ 
\item Measure each copy in the standard basis and maintain counts $n(b,i)$ of the number of times outcome $\ket{b, i}$ is observed for 
$b \in {0,1}$. 
\item Set $\sigma_{i} = 1$ if $n(0,i)> 0.4 p_{i}N$ and $-1$ otherwise. 
\end{enumerate}
\item Output the unit vector $\widetilde{x}$ with $\widetilde{x}_{i} = \sigma_{i} \sqrt{p_{i}}$.  
\end{enumerate}
\end{algorithmic}
\end{algorithm}
\noindent We will use the following version of the multiplicative Chernoff bounds for the analysis. 
\begin{fact} \label{mcb} \cite{AV79} 
Let $X_{i}$ for $i \in [m]$ be independent random variables such that $X_{i} \in [0,1]$ and let $X= \sum_{i \in [m]} X_{i}$. Then, 
\begin{enumerate} 
\item For $0<\beta<1$, $\Pr [ X < (1-\beta) \E[X]) \leq e^{ -\frac{\beta^{2} \E[X]}{2}}$. 
\item For $\beta>0$, $\Pr [ X > (1+\beta) \E[X]) \leq e^{ -\frac{\beta^{2} \E[X]}{(2+\beta)}}$. 
\end{enumerate} 
Combining the two bounds for $0< \beta <1$ we have $\Pr [ |X - \E[X]| \geq \beta \E[X] ] \leq e^{ -\beta^{2} \E[X]/3}$.  
\end{fact} 
\noindent We first prove an auxiliary lemma that shows that the sign estimation procedure in Step 2 of Algorithm \ref{alg:tom} succeeds with high probability for all sufficiently large $x_{i}$.  

\begin{Lemma} \label{signs} 
Let $S:=\{ i \in [d] \;|\; x_{i}^{2} \geq \frac{\delta^{2}}{d}  \}$, then $\sigma_{i} = sgn(x_{i})$ for all $i\in S$ with probability at least $(1-1/d^{0.83})$. 
\end{Lemma} 

\begin{proof} 
Let $Z_{ij}, i\in [d], j \in N$ be the indicator random variable for the event that the $j$-th measurement outcome in step 1 of Algorithm \ref{alg:tom} is $i$
and let $Z_{i} = \sum_{j \in N} Z_{ij}$. Note that $Z_{i} =N p_{i}$ and $\E[Z_{i}]= N x_{i}^{2}$. Applying the multiplicative Chernoff bound (Fact \ref{mcb}) with $X=Z_{i}$,  
we have that $ \Pr [ |x_{i}^{2} - p_{i} | \geq \beta x_{i}^{2}   ] \leq e^{ -\beta^{2}(x_{i}^{2} N) /3 }$
for all $i \in [d]$.

Using the fact that $x_{i}^{2} N \geq 36 \ln d$ for all $i\in S$ and choosing $\beta = 1/2$ we have, 
\als{ 
\Pr [ |x_{i}^{2} - p_{i} | \geq x_{i}^{2}/2   ] \leq \frac{1}{d^{3}} 
} 
for a fixed $i \in S$. By the union bound, the event $A_1$ that $|x_{i}^{2} - p_{i} | \leq x_{i}^{2}/2$ or equivalently  $\sqrt{2p_{i} /3} \leq x_{i} \leq \sqrt{2 p_{i} }$ holds for all $i \in S$ with probability at least $(1-1/d^{2})$. 
Let us condition on this event.

We next show that the algorithm obtains the correct signs for all $i \in S$ with high probability. We provide the argument for the case when the sign is positive (i.e. when we need to output $\sigma_i=1$), the other case is similar
with $\mathbb{\E}[n(0,i)]$ replaced by $\mathbb{\E}[n(1,i)]$. 
 We have, 
\als{ 
\mathbb{\E}[n(0,i)] &= N\frac{(x_{i} + \sqrt{p_{i} } )^{2} }{4}  \geq  N\frac{ (\sqrt{2/3}+1)^{2}} {4}p_{i} \geq 0.82 p_{i} N.   
} 
Further as $i \in S$ and $p_{i} \geq x_{i}^{2}/2$ we have that $0.82 p_{i} N \geq 14.7 \ln d$.  
Using the multiplicative Chernoff bound $\Pr [ n(0,i)  \leq (1/2) \mathbb{\E}[n(0,i)]  ] \leq e^{ -\mathbb{\E}[n(0,i)] /8 }$, we 
conclude that $n(0,i) \geq 1/2 \mathbb{\E}[n(0,i)] = 0.41 p_{i} N $ with probability at least $(1-1/d^{1.83})$, and in this case $\sigma_{i}$ correctly determines the sign of $x_i$. 
By the union bound, the signs are determined correctly for all $i \in S$ with probability at least $(1- 1/d^{0.83})$, the claim follows. 
\end{proof}

\noindent The following theorem establishes the correctness of Algorithm \ref{alg:tom}. 
\begin{theorem} \label{thm:tom} 
Algorithm \ref{alg:tom} produces an estimate $\widetilde{x} \in \R^{d}$ with $\norm{\widetilde{x}}_{2}=1$ such that $\norm{\widetilde{x}  - x }_{2} \leq \sqrt{7}\delta$ with probability at least $(1-1/d^{0.83})$. 
\end{theorem} 

\begin{proof} 
As shown in the proof of Lemma \ref{signs}, the multiplicative Chernoff bound \ref{mcb} implies that 
$ \Pr [ |x_{i}^{2} - p_{i} | \geq \beta x_{i}^{2}   ] \leq e^{ -x_{i}^{2} N\beta^{2}/3 }$
for all $i \in [d]$ and for all $0< \beta <1$. Using the factorization $ |x_{i}^{2} - p_{i} | = (|x_{i}| - \sqrt{p_{i}}) (|x_{i}| + \sqrt{p_{i}})$, the Chernoff bound can be rewritten as
\als{  
\Pr [ | |x_{i}| - \sqrt{p_{i}} | \geq \beta \frac{x_{i}^{2}}{ |x_{i}| + \sqrt{p_{i}}}    ] \leq e^{ -x_{i}^{2} N\beta^{2}/3 }. 
} 
As $\sqrt{p_{i}} \geq 0$ for all $i \in [d]$ we have $\frac{\beta x_{i}^{2} }{ |x_{i}|} \geq \frac{\beta x_{i}^{2}}{ |x_{i}| + \sqrt{p_{i}}}$. It follows that $\Pr [ | |x_{i}| - \sqrt{p_{i}} | \geq \beta |x_{i}|   ] \leq e^{ -N(\beta x_{i})^{2}/3 }$ for all $i\in [d]$ and for $0< \beta < 1$. 
For $i\in S$, choosing $\beta_{i} = \frac{\delta}{ \sqrt{d} |x_{i}|}<1 $ we obtain, 
\als{ 
\Pr [ | |x_{i}| - \sqrt{p_{i}} | \geq \frac{\delta}{\sqrt{d}}    ] \leq \frac{1}{d^{12}}. 
}
By the union bound the event  $A_{2}$ that $| |x_{i}| - \sqrt{p_{i}} | \leq \frac{\delta}{\sqrt{d}}$ for all $i \in S$ occurs with probability at least $1- \frac{1}{d^{11}}$. Conditioning on $A_{2}$, we have the bound $\sum_{i \in S} (|x_{i}| - \sqrt{p_{i}})^{2} \leq \delta^{2}$. 


We can now bound the error for the algorithm conditioned on event $A_{1}$ that the signs are 
determined correctly for all $i\in S$ (Lemma \ref{signs}) and on $A_2$, and have
\all{ 
\sum_{i \in [d]} (x_{i} - \sigma(i) \sqrt{p_{i}} )^{2} &= \sum_{i \in S }  (|x_{i}| - \sqrt{p_{i} })^{2} + \sum_{i \in \overline{S} }  (|x_{i}| + \sqrt{p_{i} })^{2} \nl 
&\leq \delta^{2} + 2 \sum_{ i \in \overline{S} }  (x_{i}^{2} + p_{i}) \nl 
& \leq 3 \delta^{2} + 2 \sum_{ i \in \overline{S} } p_{i} 
} {final} 
For the second inequality, we used that $\sum_{i \in \overline{S} } x_{i}^{2}  \leq \frac{\delta^{2}}{d}. |\overline{S}| \leq \delta^{2}$. It therefore suffices to show that 
$\sum_{ i \in \overline{S} } p_{i} \leq 2\delta^{2}$ with high probability. 

Part 2 of the multiplicative Chernoff bound yields that $\Pr[ \sum_{ i \in \overline{S} } p_{i} \geq (1+ \beta)  \sum_{ i \in \overline{S} } x_{i}^{2}] \leq e^{ -\frac{\beta^{2}}{(2+\beta)}   \sum_{ i \in \overline{S} } x_{i}^{2} N}$ for all $\beta>0$. Choosing $\beta= \frac{\delta^{2}}{ \sum_{i \in \overline{S}} x_{i}^{2}} > 1$ we have,
$ 
\Pr[ \sum_{ i \in \overline{S} } p_{i} \geq  \sum_{ i \in \overline{S} } x_{i}^{2} + \delta^{2} ] \leq e^{ - 12 d \log d}
$. Thus, $\sum_{ i \in \overline{S} } p_{i} \leq  \sum_{ i \in \overline{S} } x_{i}^{2} + \delta^{2} \leq 2\delta^2$ with overwhelming probability. Substituting in equation 
\eqref{final} we obtain that with probability at least $1-1/d^{0.83}$ (ignoring lower order terms), we have $\norm{ \widetilde{x} - x}_{2}^{2} \leq 7\delta^{2}$, the theorem follows.  
\end{proof}

The success probability for our vector state tomography algorithm can be boosted to $1-1/d^{c}$ by increasing the number of samples $N$ to $Cd \ln d/\delta^{2}$ 
for suitable constants $c,C$. In the interior point method we perform tomography for $\widetilde{O}(d^{1/4})$ iterations, so Theorem \ref{thm:tom} ensures that all 
the tomography results will be correct with high probability. In order to extend this approach to all pure states instead of the sign one would need to estimate the phase $e^{i \theta_{i} } x_{i} \ket{i}$ to sufficient accuracy. 

The quantum tomography algorithm is used for learning the output of a quantum linear system solver, that is the unitary $U$ in algorithm \ref{alg:tom} is not perfect
but produces a state $\ket{\overline{x}}$ such that $\norm{ \ket{x} - \ket{\overline{x}}} \leq \epsilon$, equivalently it produces a density matrix $\overline{\rho_{x}}$ such 
that the trace-distance between the $\rho_{x}=\ket{x}\bra{x}$ 
and $\overline{\rho_{x}}$ is $O(\epsilon)$. 
As long as the error $\epsilon$ is $o(\delta^2/d)$, the trace-distance between the states $\rho_{x}^{\otimes d/\delta^2}$ and  $(\overline{ \rho_{x}})^{\otimes d/\delta^2}$ remains close to $0$ and hence any algorithm 
with input $(\overline{ \rho_{x}})^{\otimes d/\delta^2}$ will have the same guarantees as the error-free algorithm \ref{alg:tom}. 

The complexity of the linear system solver scales as $\log (1/\epsilon)$ in the error parameter by Theorem \ref{qlsa}, hence the precision can be boosted to have error $\delta^{2}/d^3$ at the cost of a logarithmic overhead in the dimension $d$. We can therefore assume that the guarantees in Theorem \ref{thm:tom} hold when the tomography algorithm is used for reconstructing the solutions to the Newton linear system. 

Note that if we want to estimate a non-unit vector $x$ and we have an error $(1\pm \delta)$ estimate $\eta_{x}$ for $\norm{x}$ and a unitary that outputs $\ket{x}$ that corresponds to the unit vector $x/\norm{x}$, then we can first use tomography to get $\widetilde{x}$ with $\norm{ \widetilde{x} - \frac{x}{\norm{x}}} \leq \delta $ and then we have that $\norm{ \eta_{x} \widetilde{x} - x} \leq 2\delta \norm{x}. $

\section{The classical interior point method} \label{cipm} 

We start by providing the details of the classical interior point method for SDPs and elements of its analysis based on \cite{BN15}. We assume the bit complexity is constant, so we hide some logarithmic factors. This method has the following complexity: (i) For Linear Programming in $n$ dimensions with $m$ constraints, the algorithm has running time $O(n^{2} (m+n)^{3/2} \log (1/\epsilon))$. (ii) For Semi-Definite Programming over $n\times n$ matrices with constraint $\sum_{k \in [m]} x_{k} A^{(k)} \succeq B$, the algorithm has running time $O(  n^{0.5} m^{3} + n^{2.5}m^{2} + n^{3.5}m)$. 

The running time of the best known method for solving SDPs \cite{LWS15} is $O(m^{3} + mn^{\omega} + m^{2}ns \log (mnR/\epsilon) )$ where the sparsity $s$ is upper bounded by $n$. The running time is still a  large polynomial in the worst case, namely $O(n^{6})$ for the case $m=O(n^{2})$. 
\subsection{Primal and dual SDPs and the central path}

We consider a pair of primal and dual SDPs having the following form, 
\al{ 
Opt(P)&= \min_{x \in \R^{m}}  \{ c^{t} x \;|\; \sum_{k \in [m]} x_{k} A^{(k)}  \succeq B \}  \nl 
Opt(D)&= \max_{Y \succeq 0 }  \{ Tr(BY) \;|\; Y \succeq 0, Tr(YA^{(j)})=c_{j}   \}
} 
We assume that the primal and dual SDPs are strictly feasible, that is thet have solutions lying in the interior of the cone of positive semi-definite matrices. 
Define $L= \Span_{k \in [m]} (A^{(k)}) $ to be the span of the matrices $A^{(k)}$, and let $L^{\perp}$ be the orthogonal complement of $L$. We assume without loss 
of generality that the matrices $A^{(k)}$ are linearly independent. Let $C$ be an arbitrary dual feasible solution, then the SDP pair above can be written in the 
following more symmetric form, 
\all{ 
Opt(P')&= \min_{S \succeq 0 }  \{ Tr(CS) + Tr(BC) \;|\; S \in (L - B)   \}  \nl 
Opt(D)&= \max_{Y \succeq 0}  \{ Tr(BY) \;|\; Y \in (L^{\perp} + C)    \}
} {sdp} 
The primal and dual objective functions are in fact symmetric as $B,C$ are constants for the primal formulation. The strict feasibility of the SDPs and the conic duality theorem  imply that strong duality holds \cite{BN15}, thus there are feasible solutions with $Opt(P')=Opt(D')$.

Let $(S, Y)$ be a pair of solutions for the primal and dual SDPs \eqref{sdp}. The duality gap $\Delta(S, Y)$ is the difference between the primal and dual objective values. 
The duality gap can be computed using the relation $Tr((S+B)(Y-C) )=0$, 
\al{ 
\Delta(S,Y)  =  Tr(CS)  +Tr(BC) - Tr(BY) = Tr( SY ) 
} 
A pair of optimal solutions $(S, Y)$ has duality gap $0$. It satisfies $Tr(SY)=0$, as $S,Y$ are positive semidefinite this also implies that $SY=YS=0$. 

The logarithmic barrier is defined as $K(X) = -\log (\det(X))$. We give expressions for the first two derivatives of the logarithmic barrier in the interior of the psd-cone. The first derivative $\nabla K(X) \in \R^{n\times n}$ while the second derivative $\nabla^{2} K(X)$ can be viewed as a function $\R^{n\times n} \to \R^{n\times n}$, it therefore suffices to evaluate $\nabla^{2} K(X)(H)$ for any $H\in \R^{n\times n}$. 
\all{ 
\nabla K(X) &= -X^{-1} \nl 
\nabla^{2} K(X) (H) &= X^{-1} H X^{-1}  
} {mder} 
The central path for a pair of primal-dual SDPs consists of the optimal solutions to the following pair of convex programs parametrized by a positive constant $\nu \in \R_{+}$. 
If the primal and dual SDPs are strictly feasible, it follows that the central path is well defined and unique \cite{BN15}. 
\all{ 
Opt(P_\nu)&= \min_{S \succeq 0 }  \{ Tr(CS) + \nu K(S) \;|\; S \in (L - B)   \}  \nl 
Opt(D_\nu)&= \max_{Y \succeq 0}  \{ Tr(BY) - \nu K(Y) \;|\; Y \in (L^{\perp} + C)    \}
}{centpath}  
The following claim characterizes the solutions lying on the central path for a given value of $\nu$. 

\begin{claim} \label{cp} 
The optimal solutions $(S_{\nu}, Y_{\nu})$ on the central path satisfy $S_{\nu} Y_{\nu} = Y_{\nu} S_{\nu} = \nu I$. 
\end{claim} 
\begin{proof}  
Let $S$ be a primal feasible solution. If the gradient of the primal objective function $C+ \nu \nabla K(S)$ evaluated at $S$ has a non zero projection onto $L$, then incrementing $S$ in the direction of the the projected gradient improves the primal objective value. As \eqref{centpath} is a convex program, it follows that $S$ is a primal optimal solution if and only if it is feasible and $C+ \nu \nabla K(S) \in L^{\perp}$. Similarly a dual solution $Y$ is optimal if and only if it is feasible and $B - \nu \nabla  K(Y) \in L$. 

Let $S_{\nu}$ be the optimal primal solution on the central path for $\nu>0$ and define $Z= \nu S_{\nu}^{-1}$. It suffices to show that $Z$ is the dual optimal solution. By the primal optimality of $S_{\nu}$ we have $C- \nu S_{\nu}^{-1} \in L^{\perp}$ which implies that $Z \in L^{\perp} + C$ is dual feasible. The gradient of the dual objective function evaluated at $Z$ is $B - \nu \nabla  K(\nu S_{\nu}^{-1} ) = B +  S_{\nu} $. By the primal feasibility of $S_{\nu}$ we have $S_{\nu} \in L- B$ which implies that $Z$ is dual optimal. 
\end{proof} 
An ideal interior point algorithm would follow the central path in the direction $\nu \to 0$. Of course, the actual method does not always follow the path, but tries to stay close to it. 
We therefore need a measure to quantify the distance of a solution pair from the central path. Define
\all{ 
d(S, Y, \nu) = \norm{ \nu^{-1} Y - S^{-1}  }_{S}  =  \norm{ \nu^{-1} S - Y^{-1}  }_{Y}
}{dist} 
where the Hessian norm is given by $\norm{X}_{Y} = \sqrt{ Tr(X^{T} (\nabla^{2} K(Y))^{-1} X )} $.
Note that $d(S, Y, \nu)=0$ if $S,Y$ lie on the central path. 
The distance can be computed using equation \eqref{mder} as, 
\al{ \label{norm}
\norm{ \nu^{-1} Y -S^{-1}  }_{S}^{2} &= Tr((\nu^{-1} Y - S^{-1})^{T}  (\nabla^{2} K(S))^{-1}   (\nu^{-1}Y - S^{-1})) \nl 
&= Tr( (\nu^{-1}Y - S^{-1}) S ( \nu^{-1}Y - S^{-1}) S) \nl 
&= Tr( (\nu^{-1} YS - I)^{2} ) = Tr( (I - \nu^{-1} S^{1/2} Y S^{1/2})^{2} ) \nl 
&= \norm{ I - \nu^{-1} S^{1/2} Y S^{1/2}} _{F} ^{2}
} 
For the computation in the second step we used the fact that for an arbitrary matrix $Z$, we have $(\nabla^{2} K(S))^{-1}  Z = S ZS$ as from equation \eqref{mder} it follows that $(\nabla^{2} K(S)) (SZS) = Z$. The following claim shows that the duality gap can be bounded in terms of the distance to the central path, in particular it shows that if $d(S, Y, \nu) \leq \eta$
for some constant $\eta \in [0,1]$, then $Tr(SY) \leq 2\nu n$. 
\begin{claim} \label{dgap} 
For all $\nu>0$, the duality gap and distance from central path are related as,  
\als{ 
 \nu ( n- \sqrt{n} d(S, Y, \nu))  \leq Tr(SY) \leq \nu ( n+ \sqrt{n} d(S, Y, \nu)) 
} 
\end{claim} 
\begin{proof} 
 Let $\lambda_{i}$ be the eigenvalues for $S^{1/2} YS^{1/2}$, then $Tr(SY) = \sum_{i\in [n]} \lambda_{i}$ and $d(S,Y, \nu) = \sum_{i \in [n]} ( 1- \nu^{-1} \lambda_{i})^{2}$
 can be related as follows, 
 \als{ 
 Tr(SY) = \sum_{i \in [n]} \lambda_{i} &\leq n \nu  + \sum_{i \in [n]} |\lambda_{i} - \nu| \nl 
 &\leq n \nu  + \sqrt{n} \en{ \sum_{i \in [n]} (\nu - \lambda_{i})^{2} } ^{1/2} \nl 
 &=  \nu   n + \nu \sqrt{n} d(S, Y, \nu)  
 } 
 Similarly,  $n \nu  \leq \sum_{i \in [n]} \lambda_{i} + \sum_{i \in [n]} |\nu - \lambda_{i}|$ by the triangle inequality and arguing as above we obtain that 
 $ \nu ( n- \sqrt{n} d(S, Y, \nu))  \leq Tr(SY)$. 
\end{proof}

\subsection{The Newton linear system} 

The interior point method starts with a pair of feasible solutions $(S,Y)$ for the primal dual SDPs \eqref{sdp} with duality gap $Tr(SY)=\nu n$ and 
$d(S, Y, \nu) \leq  \eta$ for a constant $\eta \leq 0.1$. A single step of the method updates the solution to $(S' = S + dS, Y' = Y + dY)$ such that 
$Tr(S'Y') = \nu' n$ for  
$\nu' = (1-\chi/\sqrt{n})\nu$ where $\chi \leq \eta$ is a positive constant. The updates $(S', Y')$ are computed by solving a system of linear equations called the Newton linear system which we define next. 

Let $G(S', Y', \nu') :=  S'Y' - \nu' I=0$ be an additional non-linear constraint. If this constraint is satisfied then we are on the central path $(S_{\nu'}, Y_{\nu'})$ by Lemma \ref{cp}.

The constraint $G(S', Y', \nu')=0$ can be linearized by considering the Taylor expansion of $G(S, Y, \nu')=0$ at $(S,Y)$ and setting the first order terms to $0$. 
That is, we replace $G(S', Y', \nu')=0$ by the linear constraint $dS \part{G(S, Y, \nu')}{S}  +  \part{G(S, Y, \nu')}{Y} dY + G(S, Y, \nu')= 0$ to obtain the Newton linear system, 
\all{ 
dS \in L, \;\;\; &dY \in L^{\perp}  \nl 
 dS Y + SdY &= \nu' I - SY
} {newton} 
Given a basis for $L^{\perp}$ the constraints $dS \in L$ can be written as a collection of linear equations of the form $Tr(T_{k}dS)=0$ where $T_{k}$ are a set of basis vectors 
 for $L^{\perp}$. Similarly the constraint $dY \in L^{\perp}$ can be written as a set of linear constraints given a basis for $L$. For the quantum interior point method, we have a basis for $L$ but it is computationally expensive to compute a basis for $L^{\perp}$, we will show how to solve the Newton linear system in this setting. 

The analysis of the classical interior point method shows that the Newton linear system has a unique solution. Further, the updated solution $(S'= S+ dS, Y'= Y+dY)$ is positive definite and for $\nu'= (1- \chi/\sqrt{n})\nu$ (where $\eta = \chi \leq 0.1$ are constants) the updated solution remains close to the central path, that is it satisfies $d(S', Y', \nu') \leq \eta$. It follows that $O(\sqrt{n} \log (n/\epsilon))$ iterations are required to obtain $\nu=O(\epsilon/n)$ and thus a duality gap of $\epsilon$ by Claim \ref{dgap}. 
We provide an analysis for classical interior point method in Section \ref{clan} that takes into account errors made by the quantum algorithm in solving the Newton linear system at each step of the method.

\subsection{The initial point} 
The interior point method requires an initial point that is feasible for the primal dual pair of SDPs and has a bounded duality gap. 
It may be easy to create such a pair of feasible solutions for some problems. In general, it is known that the initial point
can be found by solving Newton linear systems that are identical to the ones used for the actual interior point method.
We sketch this reduction, further details can be found in \cite{B15}. 

Consider a family of SDPs parameterized by $C$ having the form $\min_{S \succeq 0} \{ Tr(C S) \;|\; S \in L-B \}$. 
It follows from \eqref{centpath} that as $\nu \to \infty$, the central paths of all SDPs in this family converge 
to the minimizer of $K(S), S \succeq 0, S \in L-B$. The minimizer is the same for all SDPs in the family and is called the 
analytic center. Following the central path for any SDP in the family in the direction $\nu \to \infty$ converges to the analytic center.

Suppose we have some matrix $T \in L- B, T \succeq 0$. Notice that by Claim \ref{cp}, $T$ lies on the central path with $\nu=1$ 
for the primal-dual SDP pair with primal objective function $\min_{S \succeq 0} \{ Tr(T^{-1} S) \;|\; S \in L-B \}$, as $T^{-1}$ is a feasible dual solution and 
$T^{-1} T =I$.  An initial point close to the analytic center can be found 
by starting with $T$ and then following the central path for $O(\sqrt{n} \log n)$ iterations, with $\nu \to (1+ O(1/\sqrt{n})) \nu$ for each step. 

Finding such an initial point does not offer any additional difficulty, even for the quantum case as the Newton linear system depends on $(L, L^{\perp}, S, Y)$ 
but not on $(B,C)$. The quantum algorithm from Section \ref{qipm} will also be able to find an initial point and we can assume that 
we are given an initial solution $(S_{0}, Y_{0})$ such that the duality gap $Tr(S_{0} Y_{0}) = poly(n)$.

\begin{algorithm}[H]
\caption{Classical interior point method.} \label{qmat}
\begin{algorithmic}[1]
\REQUIRE Matrices $A^{(k)}$ with $k \in [m]$, $B \in \R^{n\times n}, c \in \R^{m}$ in memory, precision $\epsilon>0$.  \\
\begin{enumerate} 
\item Find feasible initial point $(S_{0}, Y_{0},\nu_0)$ close to the analytic center. 
\item Starting with $(S_{0}, Y_{0},\nu_0) $ repeat the following steps $O(\sqrt{n} \log (Tr(S_{0} Y_{0})/\epsilon))$ times. 
\begin{enumerate}
\item Solve the Newton linear system $\big[ dS \in L, dY \in L^{\perp}, S dY + dSY = (1- \frac{1}{ 10 \sqrt{n}} )\nu I - SY \big]$ to get 
$(dS, dY)$. 
\item Update $S \leftarrow S+ dS, Y \leftarrow Y + dY$, $\nu \leftarrow Tr(SY)/n$.
\end{enumerate}
\item Output $(S,Y)$.
\end{enumerate} 
\end{algorithmic}
\end{algorithm}

\section{Analysis of the approximate interior point method} \label{clan} 
We provide a convergence analysis for the approximate interior point method, which will be used for the analysis of the quantum interior point method we will describe in the next section. Our analysis follows the analysis for the exact classical interior point method in \cite{BN15}, where we extend it to the case where the Newton linear system is solved only approximately to obtain a solution
$\overline{dS \oplus dY}$ such that $\norm{ \overline{dS \oplus dY} - dS \oplus dY  }_{F} \leq \frac{\xi}{\norm{Y \oplus Y^{-1}}_{2}}$. Recall that $dS \oplus dY$ is the block diagonal matrix with blocks $dS$ and $dY$. 

Note that the above guarantee also implies that $\norm { dY - \overline{dY}  } _{F} \leq \frac{\xi}{ \norm{ Y \oplus Y^{-1} }_{2} }$ and $\norm{ dS - \overline{dS}  }_{F} \leq \frac{\xi}{ \norm{Y \oplus Y^{-1} }_{2}}$. Further, we have that $\norm{Y \oplus Y^{-1} }_{2} = \max ( \norm{Y}_{2}, \norm{Y^{-1}}_{2}) $ so we also have the 
approximation guarantees $\norm { dY - \overline{dY}  } _{F} \leq \frac{\xi}{ \norm{ Y^{-1} }_{2} }$ and $\norm{ dS - \overline{dS}  }_{F} \leq \frac{\xi}{ \norm{Y}_{2}}$. The main claim for the analysis of the approximate interior point method is as follows. 

\begin{theorem} \label{mainthm} 
Let $\chi \leq \eta = 0.1, \xi \leq 0.01$ be positive constants and let $(S,Y)$ be a pair of SDP solutions with $Tr(SY)=\nu n$ and $d(S,Y, \nu) \leq \eta$. Let $\nu'= (1 - \frac{\chi}{\sqrt{n}})\nu$, then the Newton linear system given by  
\als{ 
dS \in L, \;\;\; &dY \in L^{\perp}  \nl 
 dS Y + SdY &= \nu' I - SY
} 
has a unique solution $(dS, dY)$. Let $\norm { dY - \overline{dY}  } _{F} \leq \frac{\xi}{ \norm{Y^{-1}}_{2} } $ and $\norm{ dS - \overline{dS}  }_{F} \leq \frac{\xi}{\norm{Y}_{2}}$ be approximate solutions to the Newton linear system and let $(\overline{S}= S+\overline{dS}  , \overline{Y} = Y + \overline{dY})$ be the updated solution. Then, the following statements hold, 
\enum{ 
\item The updated solution is positive definite, that is $\overline{S} \succ 0$ and $\overline{Y} \succ 0$. 
\item The updated solution satisfies $d(\overline{S}, \overline{Y}, \overline{\nu}) \leq \eta$ and $Tr(\overline{S}\;\overline{Y}) = \overline{ \nu} n$ for $\overline{\nu} = ( 1- \frac{\alpha}{\sqrt{n}})  \nu$ for a constant $0 < \alpha \leq \chi$. 
} 
\end{theorem} 

\subsection{Proof of Theorem \ref{mainthm}} \label{proof} 
First, in the exact case, we would have $\xi=0$ and also $\alpha = \chi$. Since we now have an approximation error, the updated solution still remains close to the path, but the trace $Tr(SY)$ drops by a factor slightly less than in the exact case, since $\alpha \leq \chi$. Nevertheless, the convergence rate remains the same. 
We follow the proof in \cite{BN15} and show that it suffices to prove Theorem \ref{mainthm} for the special case where $S,Y$ are diagonal matrices, and moreover $Y=I$. Given 
approximate solutions $(\overline{dS}, \overline{dY})$ for the Newton linear system, we need to take into account the scaling of the approximation errors when we map the matrices $(S,Y)$ to this special case.

The scaling symmetry of the cone $S_{n}^{+}$ of positive semidefinite matrices corresponding to an invertible matrix $Q$ is 
given by the map $X \to QXQ$. A scaling by $Q$ when applied to the SDP pair \eqref{sdp} yields, 
\all{ 
Opt(\widehat{P})&= \min_{S \succeq 0 }  \{ Tr( Q^{-1} CQ^{-1} S  )  \;|\; S \in Q(L - B)Q   \}  \nl 
Opt(\widehat{D})&= \max_{Y \succeq 0}  \{ Tr( QBQ Y ) \;|\; Y \in Q^{-1} (L^{\perp} + C) Q^{-1}    \}
} {qsdp}  
\noindent The duality gap and the distance to the central path are invariant under scalings. Thus a solution pair $(S,Y)$ for the SDP \eqref{sdp} 
when scaled by $Q$ yields a solution $(\widehat{S}, \widehat{Y})= (QSQ, Q^{-1} Y Q^{-1})$ for the scaled SDP \eqref{qsdp} 
such that $Tr(SY)=Tr(\widehat{S} \widehat{Y})$ and $d(\widehat{S},\widehat{Y}, \mu)=d(S,Y, \mu)$. Further, the solution 
$(dS, dY)$ to the Newton linear system for SDP \eqref{sdp} when scaled by $Q$ is the solution for the Newton linear 
system for the scaled SDP \eqref{qsdp}. 

We first apply a scaling by $Q=Y^{1/2}$, the solutions for the scaled SDP are $(Y^{1/2} S Y^{1/2}, I)$
and the solutions to the scaled Newton linear system are $(\widehat{dS}, \widehat{dY}) =  ( Y^{1/2} dS  Y^{1/2} , Y^{-1/2} dY  Y^{-1/2} ) $.
We compute the approximation error for the scaled Newton linear system
using the bounds  $\norm { dY - \overline{dY}  } _{F} \leq \frac{\xi}{ \norm{Y^{-1}}_{2} }$ 
and $\norm{ dS - \overline{dS}  }_{F} \leq \frac{\xi}{\norm{Y}_{2}}$ in the statement of Theorem \ref{mainthm}.
\als{ 
\norm{\widehat{dS} - Y^{1/2} \overline{dS}  Y^{1/2} }_F = \norm{ Y^{1/2} (dS - \overline{dS} ) Y^{1/2} }_{F} = \norm{ (dS - \overline{dS} ) Y }_{F}  \leq  \norm{ dS - \overline{dS}  }_{F}  \norm{Y}_{2}  \leq  \xi  
} 
Similarly we also have that $\norm { \widehat{dY} - Y^{-1/2} \overline{dY}  Y^{-1/2} } _{F} \leq \xi$. Hence, the approximation guarantees in Theorem \ref{mainthm} imply that 
the approximation error for the scaled Newton linear system is at most $\xi$.  

We now see how to make the two matrices diagonal. Let $B$ be the eigenbasis for $Y^{1/2}SY^{1/2}$ so that the scaled matrices are diagonal in the basis $B$. 
The statements in Theorem \ref{mainthm} do not depend on the basis used to write the matrices $Y$ and $S$, we can therefore apply a basis change by $B$ to diagonalize $S$.

To sum up, similar to the exact case, it suffices to prove Theorem \ref{mainthm} under the assumptions that $S,Y$ are diagonal matrices, $Y=I$ and 
the guarantees $\norm { dY - \overline{dY}  } _{F} \leq \xi$ and $\norm{ dS - \overline{dS}  }_{F} \leq \xi$. Note that we changed notation for convenience 
and $(dS, dY)$ now represent the solutions for the scaled Newton linear system. We also further assume that $\nu=1$ for convenience, the same proof also goes through for a general $\nu$. 

Let $s_{i}$ for $i \in [n]$ be the diagonal entries of $S$. The relations $Tr(SY) = \nu n$ and 
$d(S,Y, \nu)^{2} = \norm{Y^{1/2}SY^{1/2}-\nu I}_{F}^{2}  \leq \eta^{2}$ for $\nu=1$ imply the following constraints 
on the $s_{i}$, 
\all{ 
\sum_{i \in [n] } s_{i} = n \;, \;\; 
\sum_{i \in [n]} (s_{i} - 1)^{2} \leq \eta^{2}  
} {thetabound} 
It follows that $ s_{i} \in [1-\eta, 1+\eta]$ for all $i \in [n]$. 

The Newton linear system for the scaled SDP \eqref{qsdp} has a unique solution, this also implies that the Newton linear system for the original SDP \eqref{sdp} has a unique solution. 
The proof is from \cite{BN15} and is included for completeness. 
\begin{Lemma} \label{n1} 
The Newton linear system $dS \in L,  dY \in L^{\perp}$, 
$(d S)_{ij} +  s_{i} (d Y)_{ij}  = (\nu' I - S)_{ij}$ has a unique solution. 
\end{Lemma} 
\begin{proof} 
It suffices to show that the homogeneous linear system with the right-hand side of equation $(d S)_{ij} +  s_{i} (d Y)_{ij}  = (\nu' I - S)_{ij}$ set to $0$ has only the trivial solution. 
The equation can be written as $(dY)_{ij} = -  \frac{1}{s_{i}} (dS)_{ij}$ which implies that $Tr(dY dS) = - \sum_{i,j \in [n]}  \frac{1}{s_{i}} (dS)_{ij}^{2}$. 
Further, $Tr(dY dS)=0$, since $dY \in L$ and  $dS \in L^{\perp}$. 

The coefficients $-1/s_{i}$ are non-zero as they belong to the interval $[1/(1-\eta), 1/(1+\eta)]$ for a constant $\eta \in [0,1]$, we therefore conclude that $dS =0$. It also follows that $dY=0$ as $(dY)_{ij} = - \frac{1}{s_{i}} (dS)_{ij}$. 

\end{proof} 

\noindent The following Lemma will be used for establishing the two claims in the statement of Theorem \ref{mainthm}. The proof is same as that in the classical case and 
is included for completeness.

\begin{Lemma} \label{frobbound} 
If $(dY, dS)$ are the solutions to the scaled Newton linear system for $\nu'= (1- \chi/\sqrt{n})\nu$ then $\norm{ dY }_{F} \leq \frac{ \sqrt{ \eta^{2}+ \chi^{2}}}{1-\eta}$ and  $\norm{ dS }_{F} \leq \sqrt{\eta^{2} + \chi^{2}}$. 
\end{Lemma} 
\begin{proof} 
The scaled Newton linear system has the constraints $dS \in L,  dY \in L^{\perp}$, 
$(d S)_{ij} +  s_{i} (d Y)_{ij}  = (\nu' I - S)_{ij}$. 
Multiplying the latter equation by $dY_{ij}$ and summing up over $i, j \in [n]$ we have, 
\all{ 
\sum_{ij} s_{i} (dY)_{ij}^{2} = \sum_{i}   (\nu' - s_{i} )  (dY)_{ii}
} {ysum} 
since $Tr(dSdY) =0$ due to the orthogonality of $(dS, dY)$. 
We can bound $\norm{dY}_{F}$ using the relation derived above and the fact that $(1-\eta) \leq s_{i}$ for all $ i \in [n]$, 
\all{ 
(1-\eta ) \norm{ dY }_{F}^{2} &\leq \sum_{ij} s_{i}  (dY)_{ij}^{2} = \sum_{i}  (\nu' - s_{i} )  (dY)_{ii} \nl 
&\leq \en{ \sum_{i}  (\nu' - s_{i} )^{2} }^{1/2} \en{ \sum_{i} (dY)_{ii}^{2}}^{1/2} 
} {yineq} 
The second line follows from the Cauchy-Schwarz inequality. Substituting $\nu'= 1- \chi/\sqrt{n}$ and using the bounds  $\sum_{i \in [n]} (1-s_{i} )=0$ and 
$\sum_{i} (1-s_{i} )^{2} \leq \eta^{2}$ from equation \eqref{thetabound}, 
\als{ 
 \sum_{i}  (\nu' - s_{i} )^{2} &\leq  \sum_{i}  (1- s_{i} - \chi/\sqrt{n})^{2}  =   \sum_{i}  (1- s_{i} )^{2}    + \chi^{2}   \leq \eta^{2} + \chi^{2}.   
} 
Substituting into equation \eqref{yineq}, it follows that $\norm{dY}_{F} \leq \frac { \sqrt{\eta^{2} + \chi^{2}} } { 1-\eta}$ as claimed. In order to bound $\norm{ dS}_{F}$, we use the relation $\sum_{ij}  (dS)_{ij}^{2} = \sum_{i}  (\nu' - s_{i} ) (dS)_{ii}$ analogous to equation \eqref{ysum}. Applying Cauchy-Schwarz as in equation \eqref{yineq}, we have that 
$\norm{dS}_{F}^{2} \leq \sqrt{\eta^{2} + \chi^{2}} \norm{dS}_{F} $.   
\end{proof} 
\noindent We are now ready to prove part [1.] of Theorem \ref{mainthm}, namely that the updated solutions $\overline{S}, \overline{Y}$ are positive definite.

\begin{Lemma} \label{l2} 
The matrices $\overline{Y} = I + \overline{dY} , \overline{S} =S + \overline{dS} $ are positive definite for parameters $\chi \leq \eta = 0.1$ and $\xi <0.01$. 
\end{Lemma} 
\begin{proof} 
The Frobenius norm bound for $dS, dY$ proved in Lemma  \ref{frobbound} also implies the same bound on the corresponding spectral norms. 
The smallest eigenvalue of $(I + dY)$ is at least $1-\frac { \sqrt{\eta^{2} + \chi^{2}} } { 1-\eta} >0.84$ while that for $S  + dS$ is at least $\min s_{i} - \sqrt{\eta^{2} + \chi^{2}} \geq (1-\eta) - \sqrt{\eta^{2} + \chi^{2}} > 0.75$ where we used that $\min s_{i} \geq 1-\eta$ and $ \chi \leq \eta < 0.1$.

The additive error due to the approximations $\norm { dY - \overline{dY}  } _{F} \leq \xi$ and $\norm{ dS - \overline{dS}  }_{F} \leq \xi$ is at most $0.01$, so the 
matrices $\overline{Y}, \overline{S}$ are positive definite. 

\end{proof}

\noindent In order to prove part [2.] of Theorem \ref{mainthm}, we first show that the updated solutions are also $\eta$ close to the central path for parameters 
$ \chi \leq \eta < 0.1$. 

\begin{Lemma} \label{death} 
The distance to central path is maintained, that is $d(\overline{S}, \overline{Y}, \overline{\nu})  
< \eta$ for $\overline{\nu} = (1- \alpha/\sqrt{n})\nu $, for any $0 < \alpha \leq 0.1$ and constants $ \chi \leq \eta = 0.1$, $\xi <0.01$ and $\nu = 1$. 
\end{Lemma} 

\begin{proof} 
The distance $d(\overline{S}, \overline{Y}, \overline{\nu}) = \norm { (\overline{\nu})^{-1}(\overline{S})^{1/2} \overline{Y} (\overline{S})^{1/2} - I}_{F}$ by definition \eqref{norm}. We can write the identity matrix as $I= (\overline{S}) ^{-1/2} (\overline{S})^{1/2}$ as $\overline{S}$ is a
positive definite matrix by Lemma \ref{l2}. It follows from Lemma  \ref{isospec} that $d(\overline{S}, \overline{Y}, \overline{\nu}) = \norm{(\overline{\nu})^{-1} \overline{S}\overline{Y} -  I }_{F}$. Further, using $I= (\overline{Y})^{-1}\overline{Y} $ we have that $d(\overline{S}, \overline{Y}, \overline{\nu})  = (\overline{\nu})^{-1} \norm{ (\overline{S}- \overline{\nu} \overline{Y}^{-1} ) \overline{Y} }_{F}$.

We have that $\norm{\overline{Y}}_2 = \norm{ I + dY }_2 \leq (1+ \rho )$ for $\rho := \frac{\sqrt{\eta^{2} + \chi^{2}}}{1-\eta}$ using 
Lemma  \ref{frobbound}. Hence, it suffices to upper bound $ \norm{ \overline{S}- \overline{\nu} \overline{Y}^{-1} }_{F}$, since using that $\norm{AB}_F \leq \norm{A}_F\norm{B}_2$, we have $ (\overline{\nu})^{-1} \norm{ (\overline{S}- \overline{\nu} \overline{Y}^{-1} ) Y^{'} }_{F} \leq \frac{1+\rho}{\overline{\nu}} \norm{ \overline{S}- \overline{\nu} Y^{-1}  }_{F}$. 

We split $Z= (\overline{S} - \overline{\nu} \overline{Y}^{-1})$ into a sum of three terms and then use the triangle  inequality to bound $\norm{Z}_{F}$. 
\al{ 
Z &=  (S + \overline{dS}  - \overline{\nu} (I+ \overline{d Y} ) ^{-1} )  \nl 
&= (S+ \overline{dS}  -\overline{\nu} I + \overline{dY} ) + (\overline{\nu}-1) \overline{dY}   + \overline{\nu} ( I - \overline{dY}  - (I + \overline{dY} )^{-1} ) \nl 
&:= Z_1 + Z_2 + Z_3
} 
We have the guarantees $\norm{ dS - \overline{dS}  }_{F} \leq \xi$ and $\norm{ dY - \overline{dY}  }_{F} \leq \xi$. 
We next bound the Frobenius norms of the individual terms $Z_1, Z_2$ and $Z_3$ in the above decomposition. 
\enum{ 
\item By the triangle inequality $\norm{Z_{1}}_{F} =\norm{(S+ dS -\overline{\nu} I + dY)}_{F} + 2\xi$, it therefore suffices to bound the Frobenius norm for 
$\widetilde{Z_{1}} := (S+ dS -\overline{\nu} I + dY)$. The entries $(\widetilde{Z_{1}})_{ij}$ can be computed explicitly using the Newton linear system constraint $(dS)_{ij}= (\overline{\nu} I - S )_{ij} - s_{i} (dY)_{ij}$ we have, 
\als{ 
(\widetilde{Z_{1}})_{ij} &= (S + dS + dY - \overline{\nu} I )_{ij} \nl 
&= (S + dY - \overline{\nu} I )_{ij}  - s_{i}  (dY)_{ij}+ (\overline{\nu} I - S )_{ij}\nl 
&= (1 - s_{i} ) (dY)_{ij} \leq \eta  (dY)_{ij}. 
}
Together with Lemma  \ref{frobbound} this implies that $\norm{Z_{1}}_{F} \leq \eta \rho + 2\xi$. 

\item $\norm{Z_{2}}_{F} \leq \frac{\chi}{\sqrt{n}} (\rho + \xi)$ using Lemma  \ref{frobbound}  and the fact $\norm{ dY - \overline{dY}  }_{F} \leq \xi$. 

\item The largest eigenvalue of $dY$ is at most $\rho$, we therefore have $\norm{ (I+ \overline{dY} )^{-1} - (I+ dY)^{-1} }_{F} \leq \frac{\xi}{1-\rho}$. 
By the triangle inequality $\norm{Z_{3}}_{F}  \leq \overline{\nu} (\norm{ ( I - dY - (I + dY)^{-1} ) }_{F} + \xi + \frac{\xi}{1-\rho})$. 
Let $\lambda_{i}$ be the eigenvalues of $dY$, then 
\al{ 
\overline{\nu} \norm{ ( I - dY - (I + dY)^{-1} ) }_{F} &= \overline{\nu} \en{ \sum_{i} \en{     (1- \lambda_{i}) - \frac{1}{ (1+\lambda_{i})}  }^{2}}^{1/2}  \nl 
&\leq \en{ \sum_{i} \frac{\lambda_{i}^{4}}{(1+ \lambda_{i})^{2}} }^{1/2}  \leq \frac{\rho}{(1-\rho)}   \en{ \sum_{i} \lambda_{i}^{2} }^{1/2}   \nl 
& \leq \frac{\rho^{2}}{(1-\rho)} 
} 
In the second line we used that $\overline{\nu} <1$ and that the maximum absolute value of $| \lambda_{i} | \leq \rho$, in the third line we used the Frobenius norm bound from Lemma  \ref{frobbound}. 
} 

Combining the three bounds we have that, 
\all{
d(\overline{S}, \overline{Y}, \overline{\nu})   \leq  \frac{ (1+\rho)}{ \overline{\nu}} \en{ \eta \rho + 2\xi +  \frac{\chi}{\sqrt{n}} (\rho + \xi) } +  (1+\rho) \en{ \frac{\rho^{2} + \xi}{(1-\rho)} + \xi } 
} {db} 
If $\chi \leq \eta = 0.1$ and $\xi < 0.01$, and for any $0 < \alpha \leq 0.1$, the right hand side is approximately $0.05 + 4.34 \xi + o(1/\sqrt{n})$ which is less than $\eta$ for large enough $n$. We therefore have that $d(\overline{S}, \overline{Y}, \overline{\nu})  \leq 0.099 < \eta$. 
\end{proof} 

We note that we have not tried to optimize the values $\eta, \chi$, and $ \xi$. One should be able to get a better bound for $\xi$. We are now ready to prove the last part of Theorem \ref{mainthm}. 
\begin{Lemma} 
For $\xi <0.01$, the updated solution satisfies $Tr(\overline{S}\overline{Y}) = ( 1- \frac{\alpha}{\sqrt{n}})  n$  for some constant $0.001 \leq \alpha \leq 0.1$. 
\end{Lemma} 
\begin{proof}

We use Lemma \ref{dgap} with $\nu=(1- 0.1/\sqrt{n})$ to obtain the upper bound,  
$$Tr(\overline{S}\overline{Y}) \leq  \en{ 1- \frac{0.1}{\sqrt{n}} }   \en{  1 + \frac{d(\overline{S}, \overline{Y}, \nu)}{\sqrt{n}} } n.$$  
The proof of Lemma \ref{death} shows that $d(\overline{S}, \overline{Y}, \nu) \leq 0.099$ so we have
$$
Tr(\overline{S}\overline{Y}) \leq  \en{ 1- \frac{0.001}{\sqrt{n}} } n
$$
We use again Lemma \ref{dgap} with $\nu=(1- 0.03/\sqrt{n})$ to get the lower bound  $Tr(\overline{S}\overline{Y}) \geq (1- \frac{0.03+ d(\overline{S}, \overline{Y}, \nu)}{\sqrt{n}})n \geq (1- \frac{0.1}{\sqrt{n}}) n$, where the last step follows by computing the upper bound on $d(\overline{S}, \overline{Y}, \nu)$ given by equation \eqref{db}. It follows that $Tr(\overline{S}\overline{Y}) = ( 1- \frac{\alpha}{\sqrt{n}})  n$ for some 
$0.001 \leq  \alpha \leq 0.1$. 
\end{proof}

\section{The quantum interior point method} \label{qipm}

In this section we present the quantum interior point method. We first provide in Section \ref{sec:fact} a factorization for the Newton linear system matrix that is used for constructing the Newton linear system. We provide the quantum interior point Algorithm \ref{alg:qipm} and its implementation in Sections \ref{sec:block} and \ref{sec:nls}. We consider the special case of linear programs in Section \ref{sec:lp}. 

\subsection{Factorizing the Newton matrix} \label{sec:fact}  

Let us fix some notation. The vectors $j$ are equal to $e_j$, i.e. the vectors with 1 at position $j$. For vectors $u,v$, let $u \otimes v$ be their tensor product and let $u \circ v$ be their concatenation. For a matrix $A \in \mathcal{R}^{n \times n}$, we denote $vec(A)$ the $n^2$-dimensional vector that corresponds to the vectorized matrix. For matrices $A,B$, when we use $A \circ B$, then we mean the vector concatenation of the vectorized matrices, and similar for $u \circ A$.   


The constraints for the Newton linear system in equation \eqref{newton} are $dS \in L, dY \in L^{\perp}, dSY + SdY = \nu^{\prime} I - SY$. 
The constraint $dY \in L^{\perp}$ can be easily written as a set of linear constraints as we have a spanning set for $L=Span_{k \in [m]} (A^{(k)})$. 
As it is computationally expensive to compute a spanning set for $L^{\perp}$, we instead introduce the variables $x_{k},dx_{k}, k \in [m]$ such that $S = \sum_{k} x_{k} A^{(k)} -B$ and $dS= \sum_{k \in [m]} dx_{k} A^{(k)} \in L$.

We will also be using that $(S,Y)$ are symmetric matrices, so in particular $S^{i} = S_{i}$ and $Y^{i} = Y_{i}$ for all $i \in [n]$. 
The following claim computes the entries the Newton linear system matrix with variables $(dx, dY)$.

\begin{claim} \label{nentry} 
The Newton linear system can be written as $M ( dx \circ dY) = ( \nu'I - SY \circ 0^{m})$ where $M \in \R^{ (m+n^{2}) \times (m+n^{2})}$ has entries explicitly given as, 
\all{ 
\left [ \begin{matrix}
(A^{(1)} Y)_{11}  &  \ldots & (A^{(m)} Y)_{11}   & (1 \otimes S_{1})^T  \\
\vdots & \vdots & \vdots & \ddots \\
(A^{(1)} Y)_{ij}  &  \ldots & (A^{(m)} Y)_{ij}  &  (j \otimes S_{i})^T   \\
\vdots &  \vdots & \vdots & \ddots \\
(A^{(1)} Y)_{nn} &  \ldots &  (A^{(m)} Y)_{nn} & ( n \otimes S_{n})^T  \\
0 &  \ldots &  0  & ( vec(A^{(1)}) )^T  \\
\vdots & \vdots & \vdots & \ddots \\
0 &  \ldots &  0  &( vec(A^{(m)}) )^T  \\
\end{matrix} \right ] 
\left [ \begin{matrix}
 dx_{1} \\
\ldots   \\
dx_{k}  \\
\ldots  \\
 dx_{m} \\  
  dY_{11} \\
\ldots   \\
dY_{ij} \\
\ldots  \\
 dY_{nn} \\   
\end{matrix} \right ] 
= \left [ \begin{matrix}
 (\nu' I - SY)_{11} \\
\vdots   \\
 (\nu' I - SY)_{ij}   \\
\vdots  \\
  (\nu' I - SY)_{nn}  \\  
 0 \\
\ldots  \\
 0 \\   
\end{matrix} \right ] 
} {newexpl} 
Note that the rows and columns of $M$ have been split into blocks of size $(m, n^{2})$ in the above equation. 
\end{claim} 
\begin{proof}

The constraints for the Newton linear system are $dS \in L, dY \in L^{\perp}, dSY + SdY = \nu^{\prime} I - SY$. As $L= Span( A^{(1)}, A^{(2)}, \ldots, A^{(m)})$, the constraint $dS \in L$ can be expressed as $dS= \sum_{k \in [m]}  dx_{k} A^{(k)} $ for scalars $dx_{k} \in \R$.
The constraint $dY \in L^{\perp}$ is equivalent to the $m$ linear equations $Tr(A^{(k)} dY) = \braket{A^{(k)}}{dY}= 0$, the last $m$ rows of the matrix $M$
in equation \eqref{newmat} represent these equations.  

The constraint $dS Y + S dY = \nu'I - SY$ can be unpacked into $n^{2}$ linear equations corresponding to the first $n^2$ rows of the matrix $M$. Let us consider the constraint corresponding to the $(i,j)$-th entry of 
$\nu'I - SY$, 
\al{ 
(\nu'I - SY)_{ij}  = (dSY)_{ij} + (SdY)_{ij}
}
For the first term we have
\al{
(dSY)_{ij} = (\sum_{k \in [m]} dx_{k} A^{(k)} Y)_{ij}   =   \sum_{k \in [m]} dx_{k} (A^{(k)} Y)_{ij} } 
which is exactly the contribution from the top-left block of $M$.
For the second term we have
\al{
(SdY)_{ij} = S_i^T \cdot dY_j = ( j \otimes S_i )^T vec(dY)
}
which is again what we have from the top-right block of $M$.
It follows that the matrix $M$ in equation \eqref{newexpl} represents the Newton linear system. 

\end{proof}

Our quantum interior point algorithm uses a factorization of the Newton matrix $M$ as a product of two matrices such that multiplication by these matrices 
can be performed efficiently given the data stored in the QRAM. In order to describe this factorization, we define the matrices  
$\widetilde{Z}$ and $\widehat{Z}$.

\begin{definition} \label{def:ylifted} 
Given $Z \in \R^{n\times n}$ the matrix $\widetilde{Z} \in \R^{n^{2} \times n^{2}}$ has rows given by $\widetilde{Z}_{ij} = ( i \otimes Z_{j})^T$ for $i, j \in [n]$.  
The matrix $\widehat{Z} \in \R^{n^{2} \times n^{2}}$ has rows given by $\widehat{Z}_{ij} = ( j \otimes Z_{i})^T$ for $i, j \in [n]$.  
\end{definition} 
\noindent Note that $\widetilde{Z}$ is block diagonal while $\widehat{Z}$ is equal to $\widetilde{Z}$ up to a permutation of rows. We next prove some useful properties of the matrices $\widetilde{Z}$ and $\widehat{Z}$ defined above.

\begin{claim} \label{ymult} 
Let symmetric matrices $Z, W \in \R^{n\times n}$.
\enum{ 
\item $\widetilde{Z} vec(W) = vec(WZ)$. 
\item $(\widetilde{Z})^{T}= \widetilde{Z^{T}}$ and $\widetilde{Z} \widetilde{W} = \widetilde{ZW}$. 
\item Let $N=\widetilde{Z}\cdot \widehat{W}$, then the $(i,j)$-th row of $N$ is $( Z_{j} \otimes W_{i})^T$ and the $(k,l)$-th column of $N$ is $( W_{l} \otimes  Z_{k})$. 
} 
\end{claim} 

\begin{proof} 
For the first part, we show that the vectors $\widetilde{Z} vec(W)$ and $vec(WZ) \in\R^{n^{2}}$ are equal on all coordinates. For a fixed coordinate $i, j \in [n]$ we have, 
$$
(\widetilde{Z} vec(W))_{ij} = \widetilde{Z}_{ij} vec(W) = 
( i \otimes Z_{j})^T \cdot vec(W) = Z_{j}^T W_{i}  = (WZ)_{ij}  
$$ 

For the second part, we note that $\widetilde{Z}$ is a block diagonal matrix with $n$ distinct blocks of size $n\times n$ each equal to $Z$. From this description of 
$\widetilde{Z}$ and $\widetilde{W}$ it is clear that $(\widetilde{Z})^{T}= \widetilde{Z^{T}}$ and $\widetilde{Z} \widetilde{W} = \widetilde{ZW}$. 

For the third part, we will explicitly compute the entries of $N=\widetilde{Z}. \widehat{W}$. The $(i,j)$-th row of $\widetilde{Z}$ is given by $( i \otimes  Z_{j})^T$ 
while the $(k,l)$-th column of $ \widehat{W}$ is given by $( W_{l} \otimes k)$. The $(i,j),(k,l)$-th entry of $\widetilde{Z}.\widehat{W}$
is therefore equal to $W_{li}Z_{jk}=W_{il}Z_{jk}=W_{li}Z_{kj}$. 
Hence, the $(i,j)$-th row of $N$ is $( Z_{j} \otimes W_{i})^T$ and the $(k,l)$-th column is $( W_{l} \otimes Z_{k})$.


 
\end{proof}

We cannot work directly with the matrix $M$ as we do not have a block encoding for it since its entries are not explicitly stored in the QRAM. 
Instead we provide a factorization $M=M_{1} M_{2}$ such that 
block encodings for $M_{1}$ and $M_{2}$ can be efficiently implemented using the data stored in the QRAM. This approach is more 
efficient that implementing directly a block encoding for $M$. The following claim provides the desired factorization for $M$. 

\begin{claim} 
Let $\mathcal{A} \in \R^{m \times n^{2}}$ be the matrix such that the rows are equal to $\mathcal{A}_{k} = vec(A^{(k)})^{T}$ for $k \in [m]$, then the Newton linear system matrix $M$ can be factorized as follows: 
\all{ 
M = M_{1} M_{2} = \left ( \begin{matrix} 
\widetilde{Y} & 0 \\ 
0 & I_{m} \end{matrix} \right ). 
\left ( \begin{matrix} 
\mathcal{A}^{T} &\widetilde{Y^{-1}} \widehat{S} \\ 
0 &\mathcal{A} \end{matrix} \right )
} {eqfact} 
\end{claim} 

\begin{proof} 
We note that $S, Y^{-1}, Y$ are symmetric matrices. 
Multiplying the two matrices $M_1$ and $M_2$ and using $\widetilde{Y}\widetilde{Y^{-1}}=I$ we get the matrix 
$$
\left ( \begin{matrix} 
\widetilde{Y}\mathcal{A}^{T} & \widehat{S} \\ 
0 &\mathcal{A} \end{matrix} \right ).
$$
It remains to show that $\widetilde{Y}\mathcal{A}^{T}$ is indeed equal to the corresponding part of $M$. By part 1 of Claim \ref{ymult} and noticing that the $k$-th column of $\mathcal{A}^T$ is equal to $vec(A^{(k)})$ we have that the $k$-th column of $\widetilde{Y}\mathcal{A}^{T}$ is equal to $\widetilde{Y} vec(A^{(k)}) = vec(A^{(k)}Y)$ as is the case for the matrix $M$. %
\end{proof}

\noindent We solve the Newton linear system by implementing block encodings for $M_{1}$ and $M_{2}$ and 
use the relation $M^{-1} = M_{2}^{-1} M_{1}^{-1}$. It is known that the matrix $M$ is invertible and there is a unique 
solution to the Newton linear system. It is also known that $Y$ is an invertible matrix 
for the interior point method, from this it follows that $M_{1}$ and $M_{2}$ are also invertible.

In the next section we show how to implement the block encoding for $M_{1}$ and $M_{2}$ required for the above Theorem and thereby obtain a quantum 
algorithm for solving the Newton linear system. 

\subsection{A quantum Newton linear system solver} \label{sec:block} 
In the previous section \ref{sec:fact}, we reduced the problem of solving the Newton linear system to implementing block encodings for the matrices $M_{1}, M_{2}$. 
In this section we show how to implement the block encodings for $M_{1}$ and $M_{2}$ required for Theorem \ref{qlsa1} and 
thus obtain a quantum algorithm for solving the Newton linear system. 

If $Z\in \R^{n\times n}$ is stored in a QRAM data structure then we can implement an efficient block encoding for the matrix $\widetilde{Z}$ with parameter $\mu(\widetilde{Z}) = \min_{p\in [0,1]} ( \sqrt{n} \norm{Z}_{F}, \sqrt{s_{2p}(Z), s_{1-2p} (Z^{T})})$ as $\widetilde{Z}$ which is a block diagonal matrix with $n$ copies of $Z$ on its diagonal blocks. 
An efficient $(\mu(\widetilde{Y}), O(\log n), 0)$ block encoding for $M_{1}$ can therefore be implemented using Theorem \ref{ds} as $Y$ is stored in the QRAM data structure. 

A block encoding for $M_{2}$ requires the ability to prepare the rows and columns of $\widetilde{Y^{-1}}$.  We compute the entries of $Y^{-1}$ classically and store them in the QRAM data structure in time $O(n^{\omega})$, where $\omega \leq 2.37$ is the matrix multiplication exponent. This pre-computation allows us to implement a 
$(\norm{M_{2}}_{F}, O(\log n), 0 )$ block encoding for $M_{2}$ in time $\widetilde{O}(1)$, that is it gives an efficient block encoding for $M_{2}$. We construct a block encoding with $\mu(M_{2}) = \norm{M_{2}}_{F}$, this can potentially be extended to the more general value for $\mu(M_{2})$ using techniques in \cite{KP17}. 

\begin{theorem} \label{nsystem} 
Let $M_{2} = \left ( \begin{matrix} 
\mathcal{A}^{T} & \widetilde{Y^{-1}} \widehat{S} \\ 
0 &\mathcal{A} \end{matrix} \right ) $, if $Y, S, A^{(k)}$ and $Y^{-1}$ are stored in QRAM data structures, then there is 
an efficient $(\norm{M_{2}}_{F}, O(\log n), 0 )$ block encoding for $M_{2}$. 
\end{theorem} 

\begin{proof} 
We give an efficient implementation the unitaries $U$ and $V$ in Lemma \ref{lepsilon} for the symmetrized matrix $\overline{M_{2}}$. In order to implement $U$ we need preparation procedures for the rows and columns of $M_{2}$ and to implement $V$ we need to know the norms of the rows and columns of $M_{2}$. 
We provide these preparation and norm computation procedures for the rows and columns of $M_{2}$. 
\enum{ 
\item The first $m$ columns and the last $m$ rows of $M_{2}$ correspond to the quantum states $\ket{A^{(k)}}$ for $k\in [m]$. These can be prepared exactly as the $A^{(k)}$ are stored in QRAM. For a matrix $Z$ stored in the QRAM, we first prepare the vector of row norms $\ket{\overline{Z}} = \frac{1}{ \norm{Z}_{F}} \sum_{ i\in [n]} \norm{Z_{i}} \ket{i}$ and then apply 
the unitary $U: \ket{i, 0} \to \ket{i, Z_{i}}$ to it, that is $U \ket{\overline{Z}, 0} = \ket{Z}$. The time required for preparing these rows and columns is $\widetilde{O}(1)$. 
The norms of these rows/columns are $\norm{A^{(k)}}_{F}$ for $k \in [m]$, these norms are known exactly. 

\item We next describe the preparation procedure for the first $n^{2}$ rows of $M$. Let $\mathcal{A}_{i,j} \in \R^{m}$ be the vector with entries 
$\mathcal{A}_{i,j,k} = (A^{(k)})_{i,j}$ for $k \in [m]$. Then, the quantum state representing the $(ni+j)$-th row of $M$ for $i, j \in [n]$ is 
$\ket{\mathcal{A}_{i,j} \circ (Y^{-1}_{j} \otimes S_{i}) }$, where we used part 3 of \ref{ymult}. In order to prepare these vectors, we provide preparation procedures for the states 
$\ket{\mathcal{A}_{i,j}} $ and $ \ket{Y^{-1}_{j}}\ket{ S_{i}}$ and then use Lemma \ref{concat} to combine the results. 

The states $\ket{\mathcal{A}_{i,j}} $ can be prepared efficiently in time $\widetilde{O}(1)$ and their norms are known exactly as $A^{(k)}, k \in [m]$ are stored in the QRAM data structure given 
by Theorem \ref{cds}. The state $\ket{Y^{-1}_{j}}\ket{  S_{i}}$ can be prepared efficiently as $Y^{-1}$ and $S$ are stored in the QRAM. Further the norm of the vector $Y^{-1}_{j} \otimes S_{i}$ is the produce of the norms of  $Y^{-1}_{j}$ and $S_i$ which are known. Applying Lemma \ref{concat} it follows that $\ket{\mathcal{A}_{i,j} \circ  (Y^{-1}_{j} \otimes  S_{i}) }$ can be prepared in time $\widetilde{O}(1)$.

\item We now describe the preparation procedure for the last $n^{2}$ columns of $M$. The quantum states corresponding to the $(nk+l)$-th column of $M$ out of the 
last $n^{2}$ columns is $ \ket{ (S_{l} \otimes Y^{-1}_{k}) \circ \mathcal{A}_{i,j}}$, where $ \mathcal{A}_{i,j}$ is as defined above and we used part 3 of Claim \ref{ymult}. 

The preparation procedure for $ \ket{ (S_{l} \otimes Y^{-1}_{k}) \circ \mathcal{A}_{i,j}}$ is analogous to the procedure in part 2. 
Also, the norms of $\mathcal{A}_{i,j}$ and $S_{l} \otimes Y^{-1}_{k}$ 
are known, so Lemma \ref{concat} can be used to prepare $\ket{\mathcal{A}_{i,j} \circ  (S_{l} \otimes Y^{-1}_{k}) }$ in time $\widetilde{O}(1)$.

}

\noindent Given the efficient implementations of the unitaries $U$ and $V$ described above it follows that
we have an efficient $(\norm{M_{2}}_{F}, O(\log n), 0)$  block encoding for $M_{2}$.

\end{proof}

\paragraph{Finding $dS \circ dY$.}

The block encodings for $M_{1}, M_{2}$ can be used to solve the Newton linear system 
and obtain $M^{-1} \ket{ (\nu'I - SY) \circ 0^{m} }  = \ket{dx \circ dY}$. This can be used to recover $dx$ and $dY$ but it would then be expensive to construct $dS$ from $dx$. We instead give a procedure 
for transforming  $\ket{dx \circ dY}$ to  $\ket{dS \circ dY}$ before estimating the norms and performing tomography to the desired accuracy.

Define the matrix $M_{3}  \in \R^{2n^{2} \times (m+n^{2})}$ as $M_{3}  = \left ( \begin{matrix} 
\mathcal{A}^{T} & 0 \\ 
0 & I_{n^{2}} \end{matrix} \right )$, so that we have the equation $M_{3}  (dx \circ dY) = (dS \circ dY)$.  
Note that $M_{3} $ does not have full rank and is not invertible. As the $A^{(k)}, k \in [m]$ are 
stored in memory we also have an efficient $(\mu(M_{3} ), \log n, 0)$ block encoding for $M_{3} $. 
We can therefore apply the quantum transformation $M_{3}  (M_{1} M_{2})^{-1} \ket{ (\nu'I - SY) \circ 0^{m} }$ to obtain a 
state close to $\ket{dS\circ dY}$. Applying Theorem 
\ref{qlsa1} we obtain a Newton linear system solver with the following guarantees. 
\begin{theorem} \label{qnlss}  
There is a quantum algorithm that given $ \ket{( \nu'I - SY) \circ 0^{m} }$, outputs a state $\epsilon$-close to $\ket{dS \circ dY}$ in time 
$\widetilde{O}( \kappa(M_{3}  M^{-1}) (\mu(M_{1}) +  \mu(M_{2}) + \mu(M_{3}  ) ) \log (n/\epsilon))$ and a relative error $\epsilon$-estimate for $\norm{dS \circ dY}$ 
in time $\widetilde{O}( \frac{1}{\epsilon}\kappa(M_{3} M^{-1}) (\mu(M_{1}) +  \mu(M_{2}) + \mu(M_{3}  ) ) \log (n/\epsilon))$.
\end{theorem}

\subsection{A quantum Interior Point method for SDPs} \label{sec:nls}

The quantum interior point method for SDPs is presented as Algorithm \ref{alg:qipm}. It has two parameters, the number of iterations $T$ and the accuracy $\delta$ 
which will be determined using the analysis in Section~\ref{clan}. In this section, we describe in more detail the implementation of the different steps of Algorithm \ref{alg:qipm} and 
bound the running time.

\begin{algorithm} 
\caption{Quantum interior point method.} \label{alg:qipm} 
\begin{algorithmic}[1]
\REQUIRE Matrices $A^{(k)}$ with $ k\in [m], B \in \R^{n\times n}$, vector $ c \in \R^{m}$ in QRAM, parameters $T, \delta>0$.  \\
\begin{enumerate} 
\item Find feasible initial point $(S,Y, \nu)=(S_{0}, Y_{0}, \nu_0)$ and store the solution in the QRAM.
\item Repeat the following steps for $T$ iterations. 
\begin{enumerate}
\item Compute matrices $Y^{-1}$ and $\nu I - SY$ classically and store in QRAM data structure.

\hspace{-1.1cm} {\bf Estimate $dS \circ dY$}\\

\item {\it Estimate norm of $dS \circ dY$.}  

Solve the Newton linear system using block encodings for $M_{1}, M_{2}$ and $M_{3} $ (Theorem \ref{qnlss} ) to find estimate $\overline{\norm{dS \circ dY}}$ such that 
with probability $1-1/poly(n)$, 
$$ | \overline{\norm{dS \circ dY}} - \norm{dS \circ dY} |  \leq \delta \norm{dS \circ dY}.$$

\item {\em Estimate $dS \circ dY$.}

Let $U_N$ the procedure that solves the Newton linear system using block encodings for $M_{1}, M_{2}$ and $M_{3} $ to produce states $\ket{dS \circ dY}$ to accuracy $\delta^2/n^3$ (Theorem \ref{qnlss}).

Perform vector state tomography with $U_N$ (Algorithm \ref{alg:tom}) and use the norm estimate from (b) to obtain the classical estimate $\overline{dY\circ dS}$ such that with 
probability $1-1/poly(n)$, 
$$\norm{ \overline{dS \circ dY} - dS \circ dY }_{2} \leq 2\delta \norm{dS \circ dY}_2.$$

\hspace{-1.1cm} {\bf Update solution}\\
\item Update $Y \leftarrow Y +  \overline{dY}$ and $S \leftarrow S+ \overline{dS}$ and store in QRAM. \\Update $\nu \leftarrow Tr(SY)/n$. 
\end{enumerate} 
\item Output $(S,Y)$.
\end{enumerate}
\end{algorithmic}
\end{algorithm}

Step (2a) consists of a classical matrix multiplication and inversion, it can be carried out in $O(n^{\omega})$ time where $\omega \leq 2.37$ is the matrix multiplication exponent. We note that in practice the time for this step is $O(n^{3})$ since simpler methods are commonly used in classical linear system solvers.

Step (2b) requires the preparation of the input state $\ket{ (\nu I - SY) \circ 0}$ that can be created in time $\widetilde{O}(1)$ as $\nu I - SY$ is stored in the QRAM. Then, we use Theorem \ref{qnlss} which in turn uses the block encodings for $M_{1}, M_2$ and $M_{3}$ given in Section \ref{sec:block}, to find the norm estimate. Let $T_N = \widetilde{O}( \kappa(M_{3}  M^{-1}) (\mu(M_{1}) +  \mu(M_{2}) + \mu(M_{3}) ) \log (n/\delta))$ be the running time for solving the Newton linear system, the time required for finding the norm estimate to relative error $\delta$ is $\frac{1}{\delta } T_{N}$.

Step (2c) invokes the vector state tomography Algorithm \ref{alg:tom} to reconstruct with high probability a vector such that $\norm{dS\circ dY - \overline{dS \circ dY} }_{2} \leq 2 \delta \norm{dS \circ dY}_{2}$ in time $\widetilde{O}(\frac{n^{2}}{\delta^{2}}T_N)$ by using also our estimate of the norm from Step (2b).  This time subsumes the running time of Step (2b). Note that we need that the error in $U_N$ is $o( \delta^2/n^2)$ and hence we fixed it to $\delta^2/n^3$. Since this error only appears inside a logarithm this does not affect the running time. 
Overall, the running time for producing a classical estimate of $dS\circ dY$ in Step (2c) is $\widetilde{O}(\frac{n^{2}}{\delta^{2}}\kappa(M_{3}M^{-1}) (\mu(M_{1}) +  \mu(M_{2}) + \mu(M_{3})) \log (n/\delta))$. Hence, we have the following corollary
\begin{corollary}\label{onestep}
One iteration of the quantum interior point method for SDPs has running time $\widetilde{O}(n^\omega + \frac{n^{2}}{\delta^{2}}\kappa(M_{3}M^{-1}) (\mu(M_{1}) +  \mu(M_{2}) + \mu(M_{3})) \log (n/\delta))$ and produces with probability $1-1/poly(n)$ an estimate to the Newton Linear System solution with
$$\norm{ \overline{dS \circ dY} - dS \circ dY }_{2} \leq \delta \norm{dS \circ dY}_2.$$ 
\end{corollary}
\noindent It remains to find the values for the parameters $T$ and $\delta$ so that the quantum interior point method converges to an approximate solution of the SDP. 

First, we make a claim that allows us to better bound the norm of the sum of errors when the tomography Algorithm \ref{alg:tom} is used in a sequence of independent trials as in the quantum interior point method. 
\begin{claim} \label{gauss} 
Let $\norm{ \overline{x}_{i} - x_{i} } = \eta_{i}$ be the approximation error for a sequence of independent applications of the tomography Algorithm \ref{alg:tom} for $i \in [m]$. 
Let $x= \sum_{i \in [m]} x_{i}$ and $\overline{x} = \sum_{i \in [m]} \overline{x_{i}}$, then with high probability, 
\als{ 
\norm{\overline{x} - x}_{2}^{2}  \leq 4 \sum_{i \in [T]} \eta_{i}^{2} 
} 
\end{claim} 
\noindent For the above claim it suffices to show that $\E[ \braket{ \overline{x}_{i} - x_{i}}{ \overline{x}_{j} - x_{j}} ]\approx 0$ for independent error vectors generated by the tomography algorithm. One can see that the error vectors 
in fact have a distribution close to a multivariate Gaussian, in which case the claim holds. We will provide a proof in the full version.\footnote{Without this claim, one can use a simple union bound which adds a $\sqrt{n}$ factor to the running time of the algorithms.}

We will now use Theorem  \ref{mainthm} to provide a bound on the number of iterations. 
The analysis of the convergence of the approximate interior point method in Section \ref{clan} requires for every iteration the approximation guarantee
$\norm{ \overline{dS \oplus dY} - dS \oplus dY  }_{F} \leq \frac{\xi}{\norm{Y \oplus Y^{-1}}_{2}}$. As $\norm{dS \circ dY}_{2} = \norm{ dS \oplus dY}_{F}$ and using Corollary \ref{onestep}, in order to obtain this guarantee the precision for the 
tomography algorithm in step (2c) must satisfy $\delta \leq  \frac{\xi}{ \norm{dS \oplus dY }_{F} \norm{Y \oplus Y^{-1}}_{2}}$. We choose 
$\delta =  \frac{\xi}{ 2 \gamma \norm{dS \oplus dY }_{F} \norm{Y \oplus Y^{-1}}_{2}}$ where $\gamma>1$. The parameter $\gamma$ depends on whether 
we want absolute or relative error guarantees for the final solution and is computed below.

\begin{theorem} \label{sdpthm} 
After $T=O(\sqrt{n} \log (n/\epsilon))$ iterations, the quantum interior point method with high probability finds a pair of positive definite matrices $(S, Y)$ such that $Tr(SY) \leq \epsilon$ and the constraints $(S, Y) \in (L- B,  L^{\perp} + C)$ are satisfied approximately in the following sense. 
\begin{enumerate} 
\item If  $ \norm{B \oplus C}_{F}  > \sqrt{T}$, then $(S, Y) \in (L- B', L^{\perp} + C')$ such that 
we have $\norm{ B \oplus C  -  B' \oplus C'  }_{F} \leq \xi \norm{B \oplus C}_{F}$ and the running time is
$$\widetilde{O}(n^{\omega+0.5}\log (n/\epsilon) + \frac{n^{2}}{\xi^{2}} \sum_{ i \in [T]}    (\mu(M_{1,i}) + \mu(M_{2,i}) + \mu(M_{3}) )  \kappa(M_{3} M_{i}^{-1}) \kappa(Y \oplus Y^{-1})^{2} \log (n/\xi) ).
$$  
\item 
We have $(S, Y) \in (L- B', L^{\perp} + C')$ with $\norm{ B \oplus C  -  B' \oplus C'  }_{F} \leq \xi $ and the running time is
$$\widetilde{O}(n^{\omega+0.5}\log (n/\epsilon) + \frac{n^{2.5}}{\xi^{2}} \sum_{ i \in [T]}    (\mu(M_{1,i}) + \mu(M_{2,i}) + \mu(M_{3}) )  \kappa(M_{3} M_{i}^{-1}) \kappa(Y \oplus Y^{-1})^{2} \log (n/\xi) ).
$$ 
  
\end{enumerate} 
\end{theorem} 
\begin{proof} 
Theorem \ref{mainthm} shows that in each iteration of the quantum interior point method, $Tr(SY)$ decreases by a multiplicative factor of $(1- \frac{\alpha}{ \sqrt{n}})$ for some 
constant $\alpha$. It follows that for the solutions $(S,Y)$ obtained after $O(\sqrt{n} \log (n/\epsilon))$ of the method we have that $Tr(SY) \leq \epsilon$.  

The solutions $(S, Y)$ found by Algorithm \ref{alg:qipm} are not exactly feasible as $(\overline{dS}_{i}, \overline{dY}_{i})$ for each step $i\in [T]$ do not belong to $L, L^{\perp}$. Consequently, $(S, Y) \in L - B', L^{\perp} + C'$ with $\norm{ B \oplus C  -  B' \oplus C'  }_{F}  = \norm{\sum_{i \in [T]} (\overline{dY_{i} \circ dS_{i} } -  dY_{i} \circ dS_{i} ) }_{2}$. 
This expression is the sum of errors for the tomography algorithm over a sequence of $T$ independent steps and can therefore be bounded using Claim \ref{gauss}. 

The error $\norm{\overline{dY_{i} \circ dS_{i} } -  dY_{i} \circ dS_{i}  }_{2}$ for the $i$th step of Algorithm \ref{alg:qipm} is at most $\frac{ \xi} {2 \gamma \norm{Y \oplus Y_{i}^{-1}}_{2} }$ with probability at least 
$1- 1/poly(n)$. Note that $\norm{Y_{i} \oplus Y_{i}^{-1}}_{2} \geq \max (\lambda_{max}(Y_{i} ), \frac{1}{ \lambda_{min}(Y_{i} )} ) \geq \frac{\lambda_{max}(Y_{i} )} { \lambda_{\min}(Y_{i} )} = \kappa(Y_{i} ) \geq 1$. We therefore have the guarantee, 
\als{ 
\norm{ B \oplus C  -  B' \oplus C'  }_{F}^{2}  \leq   \ \sum_{i \in [T]} \frac{ 4\xi^{2} } { 4\gamma^{2} \norm{Y_{i} \oplus Y_{i}^{-1}}_{2}^{2} }  \leq \frac{\xi^{2} T}{ \gamma^{2}}
} 
The approximation guarantees in parts 1 and 2 correspond to $\gamma = 1$ and $\gamma=\sqrt{T}$. 
 
We conclude by computing the running time. Substituting our value for $\delta$ in the running time of one iteration from Corollary \ref{onestep}, we have
$$\widetilde{O}(n^\omega + \frac{\gamma^2n^{2}(\norm{dS \oplus dY}_{F} \norm{Y \oplus Y^{-1}}_{2})^{2} }{\xi^{2}}\kappa(M_{3}M^{-1}) (\mu(M_{1}) +  \mu(M_{2}) + \mu(M_{3})) \log (n/\delta)).$$

We upper bound the norms that appear in the above expression.

\all{ 
\norm{dS \oplus dY}_{F} \norm{Y \oplus Y^{-1}}_{2} &= \kappa(Y \oplus Y^{-1}) \lambda_{min} (Y \oplus Y^{-1} ) \norm{dS \oplus dY}_{F} \nl 
&\leq  \kappa(Y \oplus Y^{-1}) \norm{ (Y \oplus Y^{-1}) (dS \oplus dY) }_{F} \nl 
&=  \kappa(Y \oplus Y^{-1})  \norm{ Y^{1/2} dS Y^{1/2}  \oplus Y^{-1/2} dY Y^{-1/2} }_{F} \nl 
&\leq \kappa(Y \oplus Y^{-1})/4  
} {capbound} 
We used Fact \ref{frobfact} for the first inequality, the second inequality follows from Lemma \ref{frobbound} in the classical analysis of the interior point method, 
which establishes the bounds  $\norm{ Y^{1/2} dS Y^{1/2}}_{F} < 1/6$ and $\norm{ Y^{-1/2} dY Y^{-1/2} }_{F} \leq \frac{\sqrt{ \eta^{2} + \chi^{2} }}{1-\eta} <1/6$. 

Thus, the quantum interior point method for semi-definite programs as described in Algorithm~\ref{alg:qipm} has running time 
$$\widetilde{O}(n^{\omega+0.5}\log (n/\epsilon) + \frac{\gamma^{2} n^{2}}{\xi^{2}} \sum_{ i \in [T]}    (\mu(M_{1,i}) + \mu(M_{2,i}) + \mu(M_{3}) )  \kappa(M_{3} M_{i}^{-1}) \kappa(Y \oplus Y^{-1})^{2} \log (n/\delta) ).
$$ 
Substituting $\gamma = 1$ and $\gamma=\sqrt{T}$ we obtain the running times stated in parts 1 and 2. 
\end{proof}

Let us make a number of remarks about our algorithm.

First, we note that the relative error approximation holds under the assumption that $\norm{B \oplus C}_{F} \geq \sqrt{T} = O(n^{1/4})$. If we scale the SDP so that $\norm{A^{(i)}}_{2}, \norm{B}_{2} \leq 1$ (this is also the scaling used in \cite{AGGW17}), then such an assumption would be reasonable if the matrices $(B,C)$ have high rank. In fact, if the matrix with spectral norm $1$ is roughly full rank and well-conditioned, then we expect the Frobenius norm to be $O(\sqrt{n})$. 

Second, we believe that the relative error will be sufficient in many practical cases and given that the Frobenius norm of the matrices grows with the dimension it makes more sense to talk about relative than absolute error.  

Third, another way of viewing our error guarantees is that we find $\epsilon$-optimal 
solutions to an SDP that is $\xi$-close to the original SDP that we wanted to solve. If there is some notion of robustness in the SDP, meaning that it is sufficient for the constraints to be almost satisfied, then the dependence of the running time of our algorithm to parameter $\epsilon$ is only logarithmic as for the classical interior point method. 

Fourth, the parameters $\mu$ are less than the sparsity of the matrices so for comparison we can have a linear dependence on the sparsity and a running time of $\widetilde{O}(n^{2.5}s/\xi^2)$ or $\widetilde{O}(n^3s/\xi^2)$ (for different errors) for sparse matrices. In the worst case, $\mu$ is the square root of the dimension and hence the running time becomes up to $\widetilde{O}(n^{3.5}/\xi^2)$ or $\widetilde{O}(n^4/\xi^2)$, where now this is the running time for all matrices, as long as they are well-conditioned in the sense we described above.

Last, the algorithm involves a classical computation of $Y^{-1}$ and $\nu I - SY$ in step (2a) This can be replaced by quantum estimation procedures but this does not improve the worst case running time. If we use quantum linear system solvers to create the block encoding for $M_{2}$ in Theorem \ref{nsystem}, then the block encoding requires time $\widetilde{O}(\mu(Y) \kappa(Y))$ and in fact the worst case running time increases by a factor $O(\sqrt{n})$. If we use quantum linear system solvers to construct $Y^{-1}$ 
and then perform tomography to store $Y^{-1}$ in the QRAM data structure, the tomography precision required would be $\delta^{2}/n^{3}$ leading to a prohibitively large running time. We have therefore used classical linear algebra to compute $Y^{-1}$ and $\nu I - SY$. It remains an open question if more advanced quantum methods can be used for this part of the method.

\subsection{Linear Programs} \label{sec:lp} 
We observe that the Newton linear system simplifies for the case of linear programs as it corresponds to the case where $S,Y$ are diagonal matrices 
and can be represented as vectors in $\R^{n}$. Consider a pair of primal dual LPs of the form, 
\als{ 
Opt(P) &= \min_{x \in \R^{m} } \{ c^{t} x \;|\; \sum_{i \in [m]} x_{i} a_{i} \succeq b, a_{i} \in \R^{n}\} \\
Opt(D) &= \max_{y \succeq 0} \{ b^{t} y \;|\; y^{t} a_{j} = c_{j} \} 
} 
Similar to the SDP case, we define $L= Span_{i \in [m]} (a_{i})$ and $s:=Ax-b$ and rewrite the primal dual pair of LPs in a symmetric form, 
\als{ 
Opt(P') &= \min_{s \succeq0 } \{ c^{t} s \;|\; s\in L -b \} \\
Opt(D') &= \max_{y \succeq 0} \{ b^{t} y \;|\;y \in L^{\perp} + c \} 
} 
The Newton linear system \eqref{newton} for this special case reduces to $ds \in L, dy \in L^{\perp}, ds \odot y + dy \odot s = \nu \vec{1} - s\odot y$. 
As with the SDPs we introduce variables $dx_{i}$ such that $ds = \sum_{i} dx_{i} a_{i}$, with these variables the matrix $M$ for the Newton linear system for an LP 
has dimensions $(n+m) \times (n+m)$ and has entries given by,

\all{ 
diag(y\circ 1)  \left [ \begin{matrix}
a_{11} &  \ldots & a_{m1}   & \ldots  \\
\vdots & \vdots & \vdots & \ddots \\
a_{1i} &  \ldots & a_{mi}   &   diag(y^{-1} \odot s)  \\
\vdots &  \vdots & \vdots & \ddots \\
a_{1n} &  \ldots &  a_{mn} & \ldots  \\
0 &  \ldots &  0  &a_{1}  \\
\vdots & \vdots & \vdots & \ddots \\
0 &  \ldots &  0  &a_{m}  \\
\end{matrix} \right ] 
\left [ \begin{matrix}
 dx_{1} \\
\ldots   \\
dx_{k}  \\
\ldots  \\
 dx_{m} \\  
  dy_{1} \\
\ldots   \\
dy_{i} \\
\ldots  \\
 dy_{n} \\   
\end{matrix} \right ] 
= \left [ \begin{matrix}
 (\nu I - s\odot y)_{1} \\
\vdots   \\
 (\nu I - s \odot y)_{i}   \\
\vdots  \\
  (\nu I - s \odot y )_{n}  \\  
 0 \\
\ldots  \\
 0 \\   
\end{matrix} \right ] 
} {newmat} 
As for the SDP we implement the block encoding for $M$ by factorizing it into $M_{1}$ and $M_{2}$ defined by the above equation. Analogous to Step (2a) in the SDP algorithm, the entries of the vector $y^{-1} \circ s$ and also $y\circ s$ can be computed and stored in memory in time $O(n)$ unlike the SDP case where time needed 
is $O(n^{\omega})$.

The block encoding for $M_{1}$ can be implemented efficiently as $y$ is stored in the QRAM. The rows and columns of $M_{2}$ can also be prepared easily using $y^{-1} \circ s$ and data stored in the QRAM, so it is also easy to construct an efficient block encoding for $M_{2}$. Similar to the SDP case, the 
precision for the quantum tomography algorithm for recovering $y$ can be upper bounded by $\kappa(y)=\frac{y_{max}}{y_{min}}$. 
The number of iterations $T=\widetilde{O}(\sqrt{n})$ and tomography is performed on a vector of dimension $(n+m)=O(n)$ in each iteration in time $\widetilde{O}(n \kappa(y_{i})/\xi^{2})$. We can therefore specialize Theorem \ref{sdpthm} for linear programs to obtain the following result. 

\begin{theorem} \label{lpthm} 
After $T=O(\sqrt{n} \log (n/\epsilon))$ iterations, the quantum interior point method for LPs produces with high probability a pair of solutions $(s, y)$ such that $\braket{s}{y} \leq \epsilon$ and the constraints $(s, y) \in (L-b, L^{\perp} +c)$ are satisfied approximately in the following sense,  
\begin{enumerate} 
\item If  $ \norm{b \circ c}_{2}  > \sqrt{T}$, then  $(s, y) \in (L-b', L^{\perp} +c')$ a
we have $\norm{ b \circ c  -  b' \circ c'  }_{2} \leq \xi \norm{  b \circ c }_{2}$ and the running time is
$$\widetilde{O}(n^{1.5}\log (n/\epsilon) + \frac{n}{\xi^{2}} \sum_{ i \in [T]}    (\mu(M_{1,i}) + \mu(M_{2,i}) + \mu(M_{3}) )  \kappa(M_{3} M_{i}^{-1}) \kappa(y \circ y^{-1})^{2} \log (n/\xi) ).
$$  
\item 
We have $(s, y) \in (L-b', L^{\perp} +c')$ such that $\norm{ b \circ c  -  b' \circ c'  }_{2} \leq \xi$  with running time, 
$$\widetilde{O}(n^{1.5}\log (n/\epsilon) + \frac{n^{1.5}}{\xi^{2}} \sum_{ i \in [T]}    (\mu(M_{1,i}) + \mu(M_{2,i}) + \mu(M_{3}) )  \kappa(M_{3} M_{i}^{-1}) \kappa(y \circ y^{-1})^{2} \log (n/\xi) ).
$$ 
  
\end{enumerate} 
\end{theorem}

The parameters $\mu$ is less than the sparsity of the matrices so for comparison we can have a linear dependence on the sparsity and a running time of $\widetilde{O}(n^{1.5}s/\xi^2)$ or $\widetilde{O}(n^2s/\xi^2)$ (for different errors) for sparse matrices. In the worst case, $\mu$ is the square root of the dimension and hence the running time becomes up to $\widetilde{O}(n^{2}/\xi^2)$ or $\widetilde{O}(n^{2.5}/\xi^2)$, where now this is the running time for all matrices, as long as they are well-conditioned in the sense we described above.

\textbf{Acknowledgements:} A part of this work was done while the authors were visiting the Simons Institute. We thank Ronald de Wolf for helpful comments and insightful discussions on vector state tomography. This research was supported by the grants QuantERA QuantAlgo and ANR QuBIC.

\bibliographystyle{plain} 
\bibliography{b1.bib}

\end{document}